\documentclass[12pt]{amsart}

\usepackage{amsmath,textcomp,mathrsfs,textcomp,latexsym}

\usepackage{xcolor}
\usepackage{soul}

\definecolor{yel}{rgb}{1,.9,0}
\sethlcolor{yel}

\usepackage{color}
\definecolor{red-}{rgb}{1.0,0.0,0.0}
\definecolor{green-}{rgb}{0.0,0.7,0.0}
\definecolor{brown-}{rgb}{0.9,0.6,0.0}

\usepackage{amsfonts}
\usepackage{amssymb}
\usepackage[mathscr]{euscript}
\usepackage{amsmath}
\usepackage{amsthm}
\usepackage[all]{xy}
\usepackage{hyperref}
\hypersetup{
    colorlinks=true,
    linkcolor=blue,
    citecolor=blue,
    filecolor=blue,
    urlcolor=blue
}

\newtheorem{thm}{Theorem}
\newtheorem{prop}[thm]{Proposition}
\newtheorem{cor}[thm]{Corollary}
\newtheorem{lem}[thm]{Lemma}
\theoremstyle{definition}
\newtheorem{defi}[thm]{Definition}
\newtheorem{obs}[thm]{Remark}
\newtheorem{ex}[thm]{Example}
\newtheorem{prog}{Program}

\def\Z{\mathbb{Z}}

\def\F{\mathbb{F}}
\def\P{\mathbb{P}}
\def\C{\mathscr{C}\subseteq\mathbb{F}_q^n}
\def\T{\mathcal{T}}

\usepackage{multicol}

\usepackage{amscd}
\usepackage{enumerate}

\usepackage{graphicx}

\usepackage{multirow}

\usepackage{verbatim}

\begin{document}

\title{On the duals codes of skew constacyclic codes}

\author{Alexis E. Almendras Valdebenito and Andrea Luigi Tironi}

\date{\today}

\address{{\small Departamento de Matem\'atica,
Universidad de Concepci\'on,
Casilla 160-C,
Concepci\'on, Chile}}
\email{alexisalmendras@udec.cl,\ atironi@udec.cl}


\maketitle

\begin{abstract}
Let $\mathbb{F}_q$ be a finite field with $q$ elements and denote
by $\theta : \mathbb{F}_q\to\mathbb{F}_q$ an automorphism of
$\mathbb{F}_q$. In this paper, we deal with skew constacyclic codes, that is,
linear codes of $\mathbb{F}_q^n$ which are invariant under the action of a semi-linear map
$\phi_{\alpha,\theta}:\mathbb{F}_q^n\to\mathbb{F}_q^n$, defined by $\phi_{\alpha,\theta}(a_0,...,a_{n-2}, a_{n-1}):=(\alpha \theta(a_{n-1}),\theta(a_0),...,\theta(a_{n-2}))$
for some $\alpha\in\mathbb{F}_q\setminus\{0\}$ and $n\geq 2$. In particular,
we study some algebraic and geometric properties of their dual codes and we give some consequences and research results on $1$-generator skew quasi-twisted codes and on MDS skew constacyclic codes.
\end{abstract}

\tableofcontents

\section*{Introduction}

\bigskip

Let $\mathbb{F}_q$ be a field with $q$ elements. A linear code $\mathcal{C}$
of length $n$ and dimension $k$, called an $[n,k]_q$-code, is a $k$-dimensional subspace of $\mathbb{F}_q^n$.
Moreover, an $[n,k]_q$-code $\mathcal{C}$ with minimum Hamming distance $d:=d(\mathcal{C})$ is denoted
as an $[n,k,d]_q$-code. A fundamental problem in coding theory is to optimize one of the parameters
$n, k$ or $d$ of a linear code, given the other two.
If $d_q(n,k)$ denotes the largest value of $d$ for which an
$[n,k,d]_q$-code exists, we call an $[n,k,d_q(n,k)]_q$-code simply a code with best known linear code (BKLC) parameters.
It is well known that a large number of new linear codes achieving the best known bounds
$d_q(n,k)$, in particular over
small fields, have been constructed as cyclic, constacyclic or quasi-cyclic codes,
including the corresponding ones in the non-commutative case (e.g., see \cite{AGAS}, \cite{B0} and \cite{M2}).
In this paper, by extending some results of \cite{M1} to the non-commutative case,
we study some algebraic and geometric
properties of skew constacyclic codes (Definition \ref{skew GCC}) and their duals, the latest ones being strongly related
to the minimum Hamming distance $d$ which is useful in error-correcting codes and for some decoding algorithms.
After then, as applications of the previous results,
we give some immediate consequences on $1$-generator skew quasi-twisted codes and on Maximum Distance Separable (MDS) skew constacyclic codes. 
Finally, new constructions of some linear codes, $1$-generator skew quasi-cyclic codes 
and some MDS skew constacyclic codes with the best known parameters for small values of $q$
(Tables \ref{arco:generalizadoinferior1}, \ref{Table 2} and \ref{Table 3}) are obtained in the same spirit of \cite{AGAS} and \cite{M3} 
by applying the previous theoretical results and computer programs written in MAGMA \cite{magma}.

\smallskip

The paper is organized as follows. After some basic notions and remarks, in Section \ref{Sec2} we recall the main properties of
skew constacyclic codes and we reprove in an easy way some results which will be useful in the next sections.
As a consequence of these results, in Section \ref{Sec2bis} we show some geometric properties about
generator and parity check matrices of skew constacyclic codes (Theorems \ref{GMaruta:2} and \ref{GMaruta:3}),
whose columns are composed of orbits of points in projective spaces via the action of semi-linear maps, in
line with a work of T. Maruta \cite{M1}
in the commutative case.
As an application of these facts, inspired by the works \cite{B}, \cite{AGAS} and \cite{M2} respectively,
in Section \ref{Sec4} we study the main properties of $1$-generator skew quasi-twisted codes and we show a method for lengthening
skew constacyclic codes of small length to construct in an easy way new examples of linear codes over small fields that meet
the parameters of some best known linear codes (Table \ref{arco:generalizadoinferior1}). Moreover,
we extend to the non-commutative case
the main result in \cite{M2} about $1$-generator quasi-cyclic codes (Theorem \ref{GMaruta:QT}) and
as search results, we give also some new examples of
$1$-generator skew quasi-cyclic codes with the best known distance $d$ for $q=4$ and $k=5$ (Table \ref{Table 2}). 
Finally, in Section \ref{MDS}
the existence and the construction of some MDS skew constacyclic codes (Table \ref{Table 3}) are completely characterized 
by simple algebraic conditions (Theorem \ref{teo:cuaderno2} and Corollary \ref{existence cor}).

\medskip

\noindent {\bf Note}. We add to this revised version, for completeness, an Appendix relative to some bounds for the Hamming distance of some skew codes, in particular for skew constacyclic codes. This material is an extended version of $\S 2$ of the previous version of this paper and this topic, along with some related arguments and consequences, will be the main object of an upcoming future note.

\medskip

\noindent {\bf Acknowledgements}. Both authors would like to express their gratitude to
the referees for their careful reading and for many useful suggestions
which shortened the proof and improved the presentation of some results of the paper in an earlier version. 
Moreover, during the preparation of this paper in the framework of the Project Anillo ACT 1415
PIA CONICYT, the authors were partially supported by Proyecto VRID N. 214.013.039-
1.OIN, and the first author was also partially supported by CONICYT-PCHA/Mag\'ister
Nacional a\~no 2013 - Folio: 221320380.

\section{Notation and background material}\label{Sec2}

Along all this paper, we will use the following notation.

\smallskip

\noindent Let $\F_q$ be a finite field with $q$ elements, where
$q=p^r$ for some prime  $p$ and $r\in\mathbb{Z}_{\geq 1}$. Define $\F_q^\ast:=\F_q\setminus\{ 0\}$
and take $\alpha\in\F_q^\ast$. Let $\theta$ be an automorphism of
$\F_q$, that is, $\theta (z):=z^{p^t}$ for any $z\in\mathbb{F}_q$, where $t$ is an integer such that $1\leq t\leq r$.

Let us recall here the definition of the main codes
we will treat in the next sections.

\begin{defi}[\cite{B0},\cite{B1},\cite{BL}]\label{skew GCC}
A linear code $\mathscr{C}\subseteq \mathbb{F}_q^n$ is called a
\textit{skew $(\alpha ,\theta )$-cyclic code
if $\mathscr{C}$ is invariant
under the semi-linear map}
\[ \phi_{\alpha,\theta} : (c_0,c_1\dots,c_{n-1})\mapsto (\alpha \theta(c_{n-1}),\theta(c_0),\dots,\theta(c_{n-2})). \]
Moreover, a skew $(\alpha ,id )$-cyclic code is called simply an \textit{$\alpha$-constacyclic code} and, for some fixed $\theta$,
we will call a linear code $\mathscr{C}\subseteq \mathbb{F}_q^n$ a \textit{skew constacyclic code}
if $\mathscr{C}$ is a skew $(\alpha ,\theta )$-cyclic code for some $\alpha\in \F_q^\ast$.
\end{defi}

\begin{obs}\label{obs}
According to \cite[$\S 2$]{BL}, in our situation we have $\phi_{\alpha,\theta} =
A\circ\Theta$, where
$$\Theta ((c_0,\dots,c_{n-1})):=(\theta(c_0),\dots,\theta(c_{n-1}))$$ and
$A(\vec{v}):=\vec{v}A$ \ for every $\vec{v}\in \mathbb{F}_q^n$ with
$$A:=\left( \begin{array}{ c | c c c }
0 & 1 & \cdots & 0  \\
\vdots & \vdots & \ddots & \vdots \\
0 & 0 & \cdots & 1 \\
\hline \alpha & 0 & \cdots &  0
\end{array}\right)\ .$$
\end{obs}

\noindent Let us show here some algebraic and geometric
properties of the dual codes of skew constacyclic codes.

First of all, with the purpose of giving an algebraic structure to
skew $(\alpha,\theta)$-cyclic codes, one defines a ring structure on the set
$$R:=\F_q[x;\theta]:=\{a_sx^s+...+a_1x+a_0\ | \ a_i\in\mathbb{F}_q\ \mathrm{and}\ s\in\mathbb{Z}_{\geq 0} \},$$
where the addition is defined to be the usual addition of polynomials and the multiplication is defined
by the basic rule $xa=\theta(a)x$ for any $a\in\mathbb{F}_q$, and extended to all elements of $R$ by associativity
and distributivity.

Consider now
the following one-to-one correspondence:
\begin{equation}\label{cod:pi3}
\begin{aligned}
\pi \colon \phantom{aaaaat}\F^n_q\phantom{aaaaa} &\longrightarrow \phantom{aaa} R/R(x^n-\alpha)\\
    (a_0,a_1,\dots,a_{n-1})&\longmapsto a_0+a_1x+\cdots a_{n-1}x^{n-1}\quad .\\
\end{aligned}
\end{equation}
Note that $\pi$ is an $\F_q$-linear isomorphism of left $\F_q$-modules.
So, we can identify $\F^n_q$ with $R/R(x^n-\alpha)$ and any vector
$\vec{a}=(a_0,a_1,\dots,a_{n-1})\in \F^n_q$ with the polynomial
class $\pi(\vec{a}):=\sum_{i=0}^{n-1}a_ix^i\in R/R(x^n-\alpha)$ via $\pi$.

Let $m$ be the order of $\theta$. If either $m\nmid n$ or
$\alpha\notin\F_q^\theta:=\{a\in\F_q : \theta(a)=a\}$,
we know that $R/R(x^n-\alpha)$ is not a
ring and we can not argue about its ideals, as in the commutative
case (e.g., see \cite[$\S 2.2$]{JLU}). On the other hand, when $\alpha =1$ and $m\mid n$, $R/R(x^n-\alpha)$
becomes a ring and one can construct a one-to-one
correspondence between skew cyclic codes (skew $(1,\theta)$-cyclic codes) and the ideals of
$R/R(x^n-1)$ (see, e.g., \cite{B0} and \cite{B2} for many results on this topic).
Anyway, without any condition on $m$ and $\alpha$, the set $R/R(x^n-\alpha)$ can
be always considered as a left $\F_q$-module, or a left
$R$-module.

The next two results give an equivalent definition of skew constacyclic codes and some of their
well-known properties, respectively.

\begin{thm}[see, e.g., \cite{B2} and \cite{TT1}]\label{R-module}
A nonempty subset $\mathscr{C}\subset\F^n_q$ is a skew
$(\alpha,\theta)$-cyclic code if and only if $\pi(\mathscr{C})$ is a left
$R$-submodule of the left $R$-module $R/R(x^n-\alpha)$.
\end{thm}

\begin{thm}[see, e.g., \cite{B0}, \cite{B1} and \cite{BL}]\label{GMaruta:1}
Let $\pi(\mathscr{C})$ be a left $R$-submodule of $R/R(x^n-\alpha)$,
i.e. $\mathscr{C}$ is a skew $(\alpha,\theta)$-cyclic code of
$\F_q^n$. Then there exists a unique monic polynomial $g(x)$ of minimal
degree in $R$, called the generator
polynomial of $\mathscr{C}$, such that
\begin{enumerate}
\item[$(a)$] $g(x)$ is a right divisor of $x^n-\alpha$;
\item[$(b)$] $\pi (\mathscr{C})=Rg(x)/R(x^n-\alpha)=:Rg(x)$;
\item[$(c)$] every $c(x)\in\pi (\mathscr{C})$ can be written uniquely as
$c(x)=f(x)g(x)\in R/R(x^n-\alpha)$, where $f(x)\in R$ has degree
less than or equal to $n-\text{deg }g(x)$. Moreover, the dimension
of $\mathscr{C}$ is equal to $n-\text{deg }g(x)$;
\item[$(d)$] if $g(x):=\sum\limits_{x=0}^k g_ix^i$, then $\mathscr{C}$
has a generator matrix $G$ given by
\begin{align*} G =
\begin{pmatrix}
g_0     &  g_1  &  \cdots  &  g_{k}  & 0          & 0         & \cdots  & 0  \\
0         &  \theta(g_0)     &  \theta(g_1)     &  \cdots     & \theta(g_{k}) & 0         &  \cdots & 0  \\
\vdots   &  \ddots &  \ddots  &  \ddots    & \ddots    & \ddots &  \ddots  & \vdots \\
0         &  0        &   0         &  \theta^{n-k-1}(g_0)       &  \theta^{n-k-1}(g_1)      &   \cdots & \cdots  & \theta^{n-k-1}(g_{k}) \\
\end{pmatrix}. \end{align*}
\end{enumerate}
\end{thm}

\noindent For skew $(\alpha,\theta)$-cyclic codes of
$\mathbb{F}_q^n$ with $\theta (z):=z^{p^t}$, where $q=p^r$ for some prime $p$, $r\in\mathbb{Z}_{\geq 2}$
and an integer $t$ such that $1\leq t\leq r-1$, type the following MAGMA Program to construct by the command SD$(n,a)$ the generator polynomials of all skew $(a,\theta)$-cyclic codes of
$\mathbb{F}_q^n$:

\begin{prog}\label{Prog-1}~
\begin{verbatim}
p:=...; r:=...; t:=...; F<w>:=GF(p^r);
R<X>:=TwistedPolynomials(F:q:=p^t);
\end{verbatim}

\begin{verbatim}
SD:=function(n,a)
P:=[0]; for i in [1..n-2] do
P:=P cat [0]; end for; T:= [-a] cat P cat [1]; f:=R!T;
V:=VectorSpace(F,n); dd:=[]; E:=[x : x in F | x ne 0];
S:=CartesianProduct(E,CartesianPower(F,n-1));
for s in S do
ll:=[s[1]] cat [p : p in s[2]];
if LeadingCoefficient(R!ll) eq 1 then
q,r:=Quotrem(f,R!ll); if r eq R![0] then
R!ll; dd := dd cat [R!ll]; end if; end if; end for;
return dd; end function;
\end{verbatim}
\end{prog}

\begin{obs}
In \cite[$\S 2$]{L}, A. Leroy gives an algorithm to find right divisors of $x^n-\alpha\in\mathbb{F}_q[x;\theta]$ by showing that factorizations of polynomials in $\mathbb{F}_q[x;\theta']$, where $\theta'$ is the Frobenius automorphism, are translated into factorizations in the usual polynomial ring $\mathbb{F}_q[x]$. 
For further MAGMA Programs about the research of right divisors of $x^n-\alpha$ and the construction of skew $(\alpha,\theta)$-cyclic codes, we refer also to \cite{TT1} and \cite{TT2}.
\end{obs}

\begin{ex}
In $\mathbb{F}_4^{14}$ with $\theta (z):=z^{2}$ for any $z\in\mathbb{F}_4$, by using the command SD(14,1) of Program 1 
we can see that there are $603$ different nontrivial right divisors of $x^{14}-1$, i.e. $603$
different nontrivial skew $(1,\theta)$-cyclic codes of $\mathbb{F}_4^{14}$ instead of
$25$ different nontrivial cyclic codes of $\mathbb{F}_4^{14}$ in the commutative case.
\end{ex}

The following results deal with
the (Euclidean) dual code $\mathscr{C}^\perp$ of a code $\mathscr{C}\subseteq \F_q^n$, i.e. the set of words which
are orthogonal to the code's words relatively to the Euclidean scalar product. In particular, the next theorem
was first presented and proven in very different forms
in \cite[$\S 4$]{BSU}, \cite[Theorem 8]{B2}, \cite[Theorem 1]{B1} and \cite[Theorem 6.1]{FG}, and here we
give a very short and simple proof of it.

\begin{thm}\label{GMaruta:1e}
Let $\mathscr{C}\subseteq\F_q^n$ be a linear code and take
$\alpha\in\F_q^\ast$. Then

$\mathscr{C}$ is a skew $(\alpha,\theta )$-cyclic code $\iff \mathscr{C}^\perp$ is a skew
$(\alpha^{-1},\theta )$-cyclic code.
\end{thm}

\begin{proof}
Assume that $\mathscr{C}$ is a skew $(\alpha,\theta )$-cyclic code.
By \cite[Proposition 25]{TT1} and Remark \ref{obs} we deduce that $\mathscr{C}^\perp$ is
invariant under the semi-linear map $\T : \F_q^n \to \F_q^n$ given by
$\T (c_0,...,c_{n-1}):=(\theta^{-1}(c_1),...,\theta^{-1}(c_{n-1}), \theta^{-1}(\alpha)\theta^{-1}(c_{0})),$
i.e. $\T\mathscr{C}^\perp=\mathscr{C}^\perp$.
Since $\T$ is invertible, then we have also $\T^{-1}\mathscr{C}^\perp=\mathscr{C}^\perp$.
Note that $$\T^{-1}(d_0,...,d_{n-1})=(\alpha^{-1}\theta (d_{n-1}),\theta (d_{0}),...,\theta(d_{n-2}))\ .$$
Thus $\T^{-1}=\phi_{\alpha^{-1},\theta}$ and by Definition \ref{skew GCC} we conclude that
$\mathscr{C}^\perp$ is a skew $(\alpha^{-1},\theta )$-cyclic code.
Finally, having in mind that $(\mathscr{C}^\perp)^\perp
=\mathscr{C}$, the converse is
immediate. \end{proof}

\begin{cor}[see also $\S\S 5,6$ of \cite{BSU}, Proposition 13 of \cite{B2} and Proposition 1 of \cite{B3}]
Let $\C$ be a skew $(\alpha,\theta)$-cyclic code. If $\{\vec{0}\}\neq\mathscr{C} \subseteq \mathscr{C}^\perp$ $($or $\{\vec{0}\}\neq\mathscr{C}^\perp \subseteq \mathscr{C})$, then
$\alpha=\pm 1$ and $\dim\mathscr{C}\leq\frac{n}{2}\leq\dim\mathscr{C}^\perp$ $($or $\dim\mathscr{C}^\perp\leq\frac{n}{2}\leq\dim\mathscr{C})$.
In particular, if $\mathscr{C}=\mathscr{C}^\perp$ then $\alpha=\pm 1$, $n$ is even and $\dim\mathscr{C}=\frac{n}{2}$.
\end{cor}
\begin{proof}
Since $\mathscr{C}=(\mathscr{C}^\perp)^\perp$, without loss of generality, we can assume that $\mathscr{C} \subseteq \mathscr{C}^\perp$.
First of all, observe that $n=\dim\mathscr{C}+\dim\mathscr{C}^{\perp}$. Hence $n\geq 2\dim\mathscr{C}$ and $n\leq 2\dim\mathscr{C}^\perp$. Finally,
by Theorems \ref{GMaruta:1} and \ref{GMaruta:1e}, let $g(x),h(x)$ be the generator polynomials of $\mathscr{C}$ and $\mathscr{C}^\perp$,
respectively. Thus
for some $h_1(x),h_2(x)\in R$ we have
$$h_1(x)h(x)=x^n-\alpha^{-1}=x^n-\alpha+(\alpha-\alpha^{-1})=h_2(x)g(x)+(\alpha-\alpha^{-1})\ .$$
Since $\mathscr{C} \subseteq \mathscr{C}^\perp$, we see also that there exists $q(x)\in R$ such that $g(x)=q(x)h(x)$.
Thus $(h_1(x)-h_2(x)q(x))h(x)=\alpha-\alpha^{-1}\in\mathbb{F}_q$. If $\alpha-\alpha^{-1}\neq 0$ then $\mathscr{C}^\perp=\mathbb{F}_q^n$,
but this gives the contradiction $\mathscr{C}=(\mathscr{C}^\perp)^\perp=(\mathbb{F}_q^n)^\perp=\{\vec{0}\}$.
So, we get $\alpha = \alpha^{-1}$ and this implies $\alpha^2=1$, i.e. $\alpha=\pm 1$.
\end{proof}

With the next result (see also \cite[Theorem 8]{B2} and \cite[Theorem 6.1]{FG}), one can write directly the generator polynomial of the dual code of any skew constacyclic code.

\begin{prop}\label{asterisco}
Let $\C$ be a skew $(\alpha,\theta)$-cyclic code generated by $g(x)=g_0+g_1x+\cdots
+g_{n-k-1}x^{n-k-1}+x^{n-k}$ with $\deg g(x):=n-k>0$. Then the dual code $\mathscr{C}^\perp$ is generated by $$h(x):=\theta^k(\hbar_0^{-1})\left[\sum_{i=1}^k \theta^i\left( \hbar_{k-i}\right)x^i+1\right]\ ,$$ where $\hbar(x)=\hbar_0+\cdots+\hbar_{k-1}x^{k-1}+x^k$ is such that $x^n-\theta^{-k}(\alpha)=g(x)\hbar(x).$
\end{prop}

\begin{proof}
From Theorem
\ref{GMaruta:1e} it follows that the dual code $\mathscr{C}^\perp$ is a skew $(\alpha^{-1},\theta)$-cyclic code. Then by Theorem \ref{GMaruta:1} we know that there exists a unique monic polynomial $h(x)$ of minimal degree in $R$ which is the generator polynomial of $\mathscr{C}^\perp$. So by \cite[Theorem 8]{B2} we see that there exist $\hbar(x)\in R$ and $c\in\mathbb{F}_q^*$ such that
$$x^n-c=g(x)\hbar(x)$$ and $h(x)=\theta^k(\hbar_0^{-1})g^\perp(x)$, where $g^\perp(x):=\sum_{i=0}^k \theta^i\left( \hbar_{k-i}\right)x^i$ and $\hbar(x)=\hbar_0+\cdots+\hbar_{k-1}x^{k-1}+x^k$. Since $g(x)\in R$ is monic and $x^n-\alpha =t(x)g(x)$ for some monic $t(x)\in R$, by \cite[Lemma 2]{B1} we obtain also that $$x^n-\theta^{-k}(\alpha)=g(x)s(x)$$ for some $s(x)\in R$. Hence
$c-\theta^{-k}(\alpha)=g(x)(s(x)-\hbar(x))$. Since $\deg g(x)>0$, this shows that $c=\theta^{-k}(\alpha)$.
\end{proof}

\section{Generator and parity check matrices}\label{Sec2bis}

We will show here a geometric property of
the dual codes of skew constacyclic codes. More precisely, we
will see first that the columns of a parity check matrix of a skew
constacyclic code can be considered as
points in a projective space which are particular orbits under the action of
a semi-linear map associated to the code. 

Denote by $GL(k,q)$ the set of $k\times k$ invertible matrices defined over $\F_q$.
Therefore, in line with
\cite{M1}, in the non-commutative case we obtain the following two results.

\begin{thm}\label{GMaruta:2}
Let $\mathscr{C}$ be a linear $[n,n-k]_q$-code. Then
$\mathscr{C}$ is a skew $(\alpha,\theta)$-cyclic code if and only if for some $\vec{u}\in\F_q^k$ and
$T\in GL(k,q)$,
$\mathscr{C}$ has a parity check matrix of the form
$$[{}^t \vec{u}, {}^t \tau (\vec{u}), {}^t \tau ^2(\vec{u}),\dots, {}^t \tau ^{n-1}(\vec{u})]$$
such that $\tau^n(\vec{u})=\alpha \vec{u}$, where $\tau (\vec{v}):=\Theta (\vec{v})T$ for every $\vec{v}\in\mathbb{F}^k_q$.
\end{thm}

\begin{proof}
Let $[{}^t \! P_1,{}^t \! P_2,\dots,{}^t \! P_n]$ be a parity check matrix of
$\mathscr{C}$, where $P_i\in\F_q^k$, and suppose that $\mathscr{C}$ is a skew $(\alpha,\theta)$-cyclic code. Then there exists 
$S\in\mathrm{GL}(k,q)$ such that $S[\alpha^{-1} {}^t \Theta(P_n),{}^t \Theta(P_1),\dots,{}^t \Theta(P_{n-1})] = [{}^t \! P_1,{}^t \! P_2,\dots,{}^t \! P_n] \ .$
This matrix equality gives

\smallskip

\noindent $P_1=\alpha^{-1}\Theta (P_n){}^t \! S \ , \ P_2=\Theta (P_1){}^t \! S \ , \ P_3=\Theta (P_2){}^t \! S = \Theta^2 (P_1){}^t \! S_\theta {}^t \! S \ , \dots $

\noindent $ P_n=\Theta (P_{n-1}){}^t \! S = \dots = \Theta^{n-1} (P_1){}^t \! S_{\theta^{n-2}} \cdots {}^t \! S_{\theta^2} {}^t \! S_\theta {}^t \! S \ ,$

\smallskip

\noindent where $M_{\theta^h}:=[\theta^h(m_{ij})]$ for every $h\in\mathbb{Z}_{\geq 0}$ and any matrix $M:=[m_{ij}]$. 
Write $\vec{u}:=P_1$ and $T:={}^t \! S$. Defining $\tau (\vec{v}):=\Theta (\vec{v})T$ for every $\vec{v}\in\mathbb{F}^k_q$, we get

\smallskip

\noindent ${}^t \! P_1={}^t \vec{u}\ ,$

\noindent ${}^t \! P_2 = {}^t \! (\Theta (P_1){}^t \! S) = {}^t \! (\Theta (\vec{u})T)={}^t \tau (\vec{u}) \ ,$

\noindent ${}^t \! P_3 = {}^t \! (\Theta^2 (\vec{u}) T_\theta T) =
{}^t \! (\Theta (\Theta (\vec{u})T)T) = {}^t \! (\Theta (\tau (\vec{u}))T) = {}^t \tau^2 (\vec{u}) \ , \  ... $

\noindent ${}^t \! P_n={}^t \! (\Theta^{n-1} (\vec{u})T_{\theta^{n-2}} \cdots T_{\theta^2} T_\theta T) = ... = {}^t \tau ^{n-1}(\vec{u})$

\noindent and $\alpha\vec{u}=\alpha P_1=\Theta (P_n){}^t \! S=\tau(P_n)=\tau^n(\vec{u})$.
Conversely, assume that $\mathscr{C}$ has a parity check matrix $H$ of the form
$$[{}^t \vec{u}, {}^t \tau (\vec{u}), {}^t \tau ^2(\vec{u}),\dots, {}^t \tau ^{n-1}(\vec{u})]=
\left( \begin{array}{ c }
\vec{h}_1  \\
\hline\vdots \\
\hline \vec{h}_k
\end{array}\right)$$ 
such that $\tau^n(\vec{u})=\alpha \vec{u}$, where $\tau (\vec{v}):=\Theta (\vec{v})T$ for every $\vec{v}\in\mathbb{F}^k_q$
and $\vec{h}_j\in\F_q^n$ for $j=1,\dots , k$. Since 
$${}^t\vec{u}=\alpha^{-1}\ {}^t\tau^n(\vec{u})=\alpha^{-1}\ {}^t\left(\Theta\left(\tau^{n-1}(\vec{u})\right)T\right)={}^tT\ \alpha^{-1}\ {}^t\Theta\left(\tau^{n-1}(\vec{u})\right)\ ,$$
$${}^t\tau(\vec{u})={}^tT\ {}^t\Theta(\vec{u})\quad \mathrm{and}\quad 
{}^t\tau^i(\vec{u})={}^tT\ {}^t\Theta\left(\tau^{i-1}(\vec{u})\right)
\ \ \mathrm{for}\ i=2,\dots , n-1\ ,$$
we have
$$H=[{}^t \vec{u}, \dots, {}^t \tau ^{n-1}(\vec{u})]={}^t T[{}^t \vec{u}, \dots, {}^t \tau ^{n-1}(\vec{u})]_\theta A={}^t T\ H_\theta A=
{}^t T\left( {}^t T_\theta\ H_{\theta^2} A_\theta\right)A=\dots $$
$$={}^t T{}^t T_\theta\dots {}^t T_{\theta^s}\ [{}^t \vec{u}, \dots, {}^t \tau ^{n-1}(\vec{u})]_{\theta^{s+1}}\ A_{\theta^s}\dots A_\theta A
={}^t T{}^t T_\theta\dots {}^t T_{\theta^s}\ H_{\theta^{s+1}}\ A_{\theta^s}\dots A_\theta A\ $$
for any $s\in\mathbb{Z}_{\geq 0}$, where $A$ is an $n\times n$ matrix as in Remark \ref{obs} with $\alpha^{-1}$ instead of $\alpha$.
Moreover, writing $\phi :=\phi_{\alpha^{-1},\theta}=A\circ\Theta$, note that
$$H= {}^t T{}^t T_\theta\dots {}^t T_{\theta^s}\ \left(H_{\theta^{s+1}}\ A_{\theta^s}\dots A_\theta A\right)
={}^t T{}^t T_\theta\dots {}^t T_{\theta^s}\
\left( \begin{array}{ c }
\phi^{s+1}\left(\vec{h}_1\right)  \\
\hline\vdots \\
\hline \phi^{s+1}\left(\vec{h}_k\right)
\end{array}\right)\ .$$
Since ${}^t T{}^t T_\theta\dots {}^t T_{\theta^s}\in\mathrm{GL}(k,q)$, we see that $\phi^{j}(\vec{h}_i)\in\mathscr{C}^\perp$
for every $i=1,...,k$ and for any $j\in\mathbb{Z}_{\geq 0}$. Hence $\mathscr{C}^\perp$ is a skew $(\alpha^{-1},\theta)$-cyclic code
and we conclude by Theorem \ref{GMaruta:1e}.
\end{proof}

 \begin{thm}\label{GMaruta:3}
Let $g(x)=a_0+a_1x+...+a_{k-1}x^{k-1}+x^k$
be a monic skew polynomial of degree $k$ in $\F_q[x;\theta]$ that divides on the right
$x^n-\alpha$, where $\alpha\in\F_q^\ast$. Then
$\mathscr{C}\subset\F_q^n$ is a skew $(\alpha,\theta)$-cyclic
$[n,n-k]_q$-code with generator polynomial $g(x)$ if and
only if $\mathscr{C}$ has a parity check matrix
$$[{}^t \! P,{}^t \! \tau (P),{}^t \! \tau^2 (P),\dots,{}^t \! \tau^{n-1}(P)],$$ where
$P=(1,0,\dots,0)\in\F_q^k$, $\tau (\vec{v}):=\Theta (\vec{v})T_{g(x)}$ for every $\vec{v}\in\mathbb{F}^k_q$
and $T_{g(x)}$ is the
companion matrix of $g(x)$, i.e.
\[ T_{g(x)}:=\left( \begin{array}{ c c c c }
0&1&\cdots &0\\
\vdots&\vdots&\ddots &\vdots\\
0&0&\cdots&1\\
-a_0&-a_1&\cdots&-a_{k-1}\\
\end{array}\right)\ . \]
\end{thm}

\begin{proof}
Consider the linear map $\pi'$ defined by
\begin{equation}\label{cod:pi4}
\begin{aligned}
\pi' \colon \phantom{aaaaat}\F^k_q\phantom{aaaaa} &\longrightarrow \phantom{aaa} R/Rg\\
    (c_0,c_1,\dots,c_{k-1})&\longmapsto c_0+c_1x+\cdots c_{k-1}x^{k-1}\ .\\
\end{aligned}
\end{equation}
Since $P=(1,0,\dots,0)$, note that
$\pi'(\tau^i (P))=x^i\pi'(P)=x^i$ for all $i\in\Z_{\geq 0}$.
Thus, we have
$$\pi'( a_0P+\cdots+a_{k-1}\tau^{k-1}(P)+\tau^k(P))=$$
$$=a_0\pi'(P)+\cdots+a_{k-1}\pi'(\tau^{k-1}(P))+\pi'(\tau^k(P))
=$$
$$=a_0 (1)+a_1(x)+\cdots+a_{k-1}(x^{k-1})+(x^k)
=g(x) = 0 \in R/Rg(x),$$
i.e.
$a_0P+a_1\tau (P)+\cdots+a_{k-1}\tau^{k-1}(P)+\tau^k(P)=(0,...,0)\in\F_q^k$.
This shows that $\pi^{-1}(g(x)):=(a_0,...,a_{k-1},1,0,...,0)$ is in the kernel $K$
of $H$, where $H$ is the matrix $[{}^t \! P,{}^t \! \tau (P),{}^t \! \tau^2 (P),\dots,{}^t \! \tau^{n-1}(P)]$
given in the statement and $\pi$ is the linear map defined by \eqref{cod:pi3}.
Moreover, as $g(x)$ is a right divisor of $x^n-\alpha$, we deduce that
$$\pi'(\tau^n(P)-\alpha P)=\pi'(\tau^n(P))-\alpha\pi'(P)=x^n-\alpha=0\in R/Rg(x),$$
and this implies that $\tau^n(P)-\alpha P=(0,...,0)\in\F_q^k$, i.e. $\tau^n(P)=\alpha P$.
By Theorem \ref{GMaruta:2} we deduce that
$K$ is a skew $(\alpha,\theta)$-cyclic code, hence by Theorem \ref{R-module}
$\pi(K)$ is a left $R$-module which contains $g(x)$. This means that $\pi^{-1}(Rg(x))\subseteq K$.
Since both spaces have dimension $n-k$, we see that $K$ is indeed the skew $(\alpha,\theta)$-cyclic code
whose generator polynomial is
$g(x)$.
\end{proof}

As an immediate consequence of Theorems \ref{GMaruta:1e}, \ref{GMaruta:2}, \ref{GMaruta:3} and
Proposition \ref{asterisco}, we have also the following characterization of a generator matrix of a skew constacyclic code.

\begin{cor}\label{cor gen}
Let $\mathscr{C}$ be a linear $[n,k]_q$-code and $\alpha\in\F_q^*$. Then
\begin{enumerate}
\item[$(a)$] $\mathscr{C}$ is a skew $(\alpha,\theta)$-cyclic code if and only if for some $\vec{a}\in\F_q^k$ and
$T\in GL(k,q)$,
$\mathscr{C}$ has a generator matrix of the form
$$[{}^t \vec{a}, {}^t \tau (\vec{a}), {}^t \tau ^2(\vec{a}),\dots, {}^t \tau ^{n-1}(\vec{a})]$$
such that $\tau^n(\vec{a})=\alpha^{-1} \vec{a}$, where $\tau (\vec{v}):=\Theta (\vec{v})T$ for every $\vec{v}\in\mathbb{F}^k_q$.
\item[$(b)$] $\mathscr{C}$ is a skew $(\alpha,\theta)$-cyclic code
with generator polynomial $g(x)=g_0+g_1x+...+g_{n-k-1}x^{n-k-1}+x^{n-k}$
if and only if $\mathscr{C}$ has a generator matrix of the form
$$[{}^t \vec{b}, {}^t \tau (\vec{b}), {}^t \tau ^2(\vec{b}),\dots, {}^t \tau ^{n-1}(\vec{b})],$$
where $\vec{b}=(1,0,...,0)$ and $\tau (\vec{v}):=\Theta (\vec{v})T_{h(x)}$ for every $\vec{v}\in\mathbb{F}^k_q$
with $h(x)$ given by {\em Proposition \ref{asterisco}}.
\end{enumerate}
\end{cor}

\begin{obs}
Comparing the two matrices in Corollary \ref{cor gen} $(a)$ and $(b)$, since they differ by an invertible matrix, it follows that
the first $k$ vectors ${}^t \vec{a}, {}^t \tau (\vec{a}), \dots, {}^t \tau ^{k-1}(\vec{a})$ generate $\mathbb{F}^k_q$ since they are linearly independent.
The same occurs for the first $k$ vectors ${}^t \vec{u}, {}^t \tau (\vec{u}), \dots, {}^t \tau ^{k-1}(\vec{u})$ of matrices as in Theorem \ref{GMaruta:2},
by comparing these matrices with those in Theorem \ref{GMaruta:3} for the same linear code.
\end{obs}

\medskip

\noindent We will give now a criterion in Proposition \ref{GMaruta:4} to determine whether two
skew $(\alpha,\theta)$-cyclic codes with some parameters are in fact
the same linear code.

\begin{defi}
A linear code $\mathscr{C}\subseteq \F_q^n$ with a parity check
matrix of the form $$[{}^t \! P, {}^t \tau(P),
{}^t\tau^2(P),\dots,{}^t \tau^{n-1}(P)]$$ with $P\in\F_q^k$,
$\tau=T\circ\Theta$, $T\in GL(k,q)$ and such that $\tau^n(P)=\alpha P$ for some $\alpha\in\mathbb{F}_q^*$ is called \textit{a code
$\mathscr{C}$ defined by $(\tau,P,n)$}. Moreover, one can define the following set:
\[ \Gamma_k^\alpha :=\{ (\tau,P,n) \mid \text{the code defined by } (\tau,P,n) \text{ is a skew $(\alpha,\theta)$-cyclic $[n,n-k]_q$-code} \}. \]
\end{defi}

\smallskip

Note that a code $\mathscr{C}$ defined
by $(\tau, P, n)$ is a skew $(\alpha,\theta)$-cyclic $[n,n-k]_q$-code if and only if the first $k$ columns of
the parity check matrix $[{}^t \! P, {}^t \tau(P),
{}^t\tau^2(P),\dots,{}^t \tau^{n-1}(P)]$ are linearly independent.
The next result shows when two codes $\mathscr{C}_i$ defined by
$(\tau_i, P_i, n)\in\Gamma_k^\alpha$ for $i=1,2$ are the same skew $(\alpha,\theta)$-cyclic code
in $\mathbb{F}_q^n$.

\begin{prop}\label{GMaruta:4}
Let $\mathscr{C}_i$ be the code defined by $(\tau_i,P_i,n)\in
\Gamma_k^\alpha$ for $i=1,2$. Then, $\mathscr{C}_1=\mathscr{C}_2$
if and only if there exists $S\in GL(k,q)$ such that
$\tau_2=\hat{S}\circ\tau_1\circ \hat{S}^{-1}$ and $\hat{S}(P_1)=P_2$, where $\hat{S}(\vec{v}):=\vec{v}S$ for every
$\vec{v}\in\mathbb{F}_q^k$.
\end{prop}

\begin{proof}
``$\Rightarrow$''\
 Since $\mathscr{C}_1=\mathscr{C}_2$, there exists $S\in GL(k,q)$ such that
\begin{align*} [{}^t \! P_2,{}^t \tau_2(P_2),{}^t \tau_2^2(P_2),\dots,{}^t \tau_2^{n-1}(P_2)] &= {}^t \! S [{}^t \! P_1,{}^t \tau_1(P),\dots,{}^t \tau_1^{n-1}(P)] \\
 & = [{}^t (P_1S),{}^t (\tau_1(P_1)S),\dots,{}^t (\tau_1^{n-1}(P_1)S)]
   \end{align*}

From the first columns, we deduce that $\hat{S}(P_1)=P_1S=P_2$, i.e.
$\hat{S}^{-1}(P_2)=P_1$. Furthermore,
$\tau_2^i(P_2)=\tau_1^i(P_1)S=\hat{S}(\tau_1^i(P_1))=\hat{S}(\tau_1^i(\hat{S}^{-1}(P_2)))=
\hat{S}\circ\tau_1^i\circ \hat{S}^{-1}(P_2)$ for $i=1,\dots,n-1$, and $\tau_j^n(P_j)=\alpha P_j$ for $j=1,2$.

Since the set $\{ P_2,\tau_2(P_2),\dots,\tau_2^{n-1}(P_2)\}$ generates $\F^k_q$,
we see that every vector $\vec{v}\in\F_q^k$
can be written as
$\vec{v}=\sum_{i=0}^{n-1}\lambda_{i} \tau_2^i(P_2)$. So
we have
$$\tau_2(\vec{v})=\sum\limits_{i=0}^{n-1}\theta(\lambda_{i})\tau_2^{i+1}\left(P_2\right)=\sum\limits_{i=0}^{n-2}\theta(\lambda_{i})\hat{S}\circ\tau_1^{i+1}\circ \hat{S}^{-1}(P_2)+
\alpha\theta(\lambda_{i})P_2\ ,$$
$$\hat{S}\circ\tau_1\circ \hat{S}^{-1}(\vec{v})
=\sum\limits_{i=0}^{n-1}\theta(\lambda_{i})\hat{S}\circ\tau_1\circ \hat{S}^{-1}\circ\tau_2^i\left(P_2\right)=
\sum\limits_{i=0}^{n-1}\theta(\lambda_{i})\hat{S}\circ\tau_1^{i+1}\circ \hat{S}^{-1}\left(P_2\right)=$$
$$= \sum\limits_{i=0}^{n-2}\theta(\lambda_{i})\hat{S}\circ\tau_1^{i+1}\circ \hat{S}^{-1}(P_2)+\alpha\theta(\lambda_{i})\hat{S}(P_1)\ ,$$
i.e. $\tau_2(\vec{v})=\hat{S}\circ\tau_1\circ \hat{S}^{-1}(\vec{v})$ for any vector $\vec{v}\in\F_q^k$. Hence $\tau_2=\hat{S}\circ\tau_1\circ \hat{S}^{-1}$.

\smallskip

\noindent ``$\Leftarrow$''\ Since there is a matrix $S\in GL(k,q)$
such that $\tau_1=\hat{S}^{-1}\circ\tau_2\circ \hat{S}$ and $P_1=\hat{S}^{-1}(P_2)$, we
see that
\begin{align*}
[{}^t \! P_2,{}^t \tau_2(P_2),\dots,{}^t \tau_2^{n-1}(P_2)] & = [{}^t (P_1S),{}^t (\tau_1(P_1)S),\dots,{}^t (\tau_1^{n-1}(P_1)S)].
\end{align*}
By Theorem \ref{GMaruta:2}, the matrices defined by
$(\tau_1,P_1,n)$ and $(\tau_2,P_2,n)$ are parity check matrices of
skew $(\alpha,\theta)$-cyclic codes which differ by an invertible matrix ${}^t S$,
i.e. they correspond to the same linear code. Hence
$\mathscr{C}_1=\mathscr{C}_2$.
\end{proof}

Finally, let $\mathscr{C}_i$ be two codes defined by
$(\tau_i, P_i, n)\in\Gamma_k^\alpha$ for $i=1,2$. Denote by $m_{\tau_i}$
the minimal polynomial of the semi-linear map $\tau_i$, that is, the monic polynomial $m_{\tau_i}$ of
minimal degree such that $m_{\tau_i}(\tau_i)=0$ for $i=1,2$ (see \cite[Proposition 3.2]{TT2}).
Let us conclude this section by showing here with an example that in the non-commutative case, i.e. when
$\theta\not= id$, we can not obtain a similar result as Theorem 4 in \cite{M1}.
More precisely,
if $\mathscr{C}_1=\mathscr{C}_2$ then it follows easily that $m_{\tau_1}=m_{\tau_2}$, since
$\tau_2=\hat{S}\circ\tau_1\circ \hat{S}^{-1}$ for some $S\in GL(k,q)$ by Proposition \ref{GMaruta:4}, but the converse
of this statement is not true in general, as the following counterexample shows.

\begin{ex}
Consider $q=4$, $g_1(x)=1+ax+x^2+x^3$ and
$g_2(x)=1+a^2 x+x^2+x^3$ in $R$, where $\theta(\not= id)$ is the Frobenius automorphism and $a$ is a root of $y^2+y+1\in\F_2[y]$. Let $\mathscr{C}_i$ be the
skew $(1,\theta)$-cyclic $[14,11]_4$-codes with generator polynomial $g_i(x)$ for $i=1,2$. Since $g_1(x)\not=
g_2(x)$ in $R/R(x^{14}-1)$, we have
$\mathscr{C}_1\not=\mathscr{C}_2$. On the other hand, we get
\begin{align*}
\tau_1&=\begin{pmatrix}
0 & 1 &0 \\
0&0&1\\
1& a &1\\
\end{pmatrix}\circ\Theta
&&,&
\tau_2&=\begin{pmatrix}
0 & 1 &0 \\
0&0&1\\
1& a^2 &1\\
\end{pmatrix}\circ\Theta
 \end{align*}
whose minimal polynomials are $m_{\tau_1}(x)=x^6+x^2+1$
and $m_{\tau_2}(x)=x^6+x^2+1$ respectively (see \cite{TT2} for their construction).
\end{ex}

\noindent
This shows clearly that in general the condition $\mathscr{C}_1=\mathscr{C}_2$ is not equivalent to
$m_{\tau_1}=m_{\tau_2}$, as
it happens
when $\theta= id$ (e.g., see Theorem 4 in \cite{M1}).

\section{Some remarks and applications}

In this last section, we will consider two main consequences of some results obtained in Section \ref{Sec2bis}.
Let us recall here that $q=p^r$ for some prime  $p$ and $r\in\mathbb{Z}_{\geq 1}$ and that $\theta$ denotes an automorphism of
$\F_q$ given by $\theta (z):=z^{p^t}$ for any $z\in\mathbb{F}_q$, where $t$ is an integer such that $1\leq t\leq r-1$.

\subsection{On $1$-generator skew quasi-twisted codes}\label{Sec4}

The main result of this section (Theorem \ref{GMaruta:QT}) is a consequence of
the previous results and it has been
motivated principally by the fact that the class of
$1$-generator skew quasi-cyclic codes
generalizes the class of
the $1$-generator quasi-cyclic codes by obtaining new examples of linear
codes with BKLC parameters (see \cite{AGAS}).
This suggests that better codes
may be found in this new class, or simply by lengthening skew constacyclic codes.

\smallskip

First of all, let us note here that
$1$-generator skew quasi-twisted (SQT) codes can be easily defined from the notion
of $1$-generator QT codes (see, \cite[$\S 1$]{M2}).
So, inspired by \cite{AGAS}, \cite{B} and \cite{M2}, we give first the following
definition of skew quasi-twisted codes.

\begin{defi}\label{skew QT codes}
Take $\alpha\in\F_q^\ast$ and let $\theta$ be an automorphism of
$\F_q$. Denote by $\hat{R}_{\alpha,\theta}:=
R/R(x^N-\alpha)$ the polynomial
ring $R$ modulo $x^N-\alpha$. For $m\in\mathbb{Z}_{\geq 1}$ and
$$\boldsymbol{g}=(g_1(x),g_2(x),\dots,g_m(x))\in \hat{R}_{\alpha,\theta}^m,$$
the set
\[ \boldsymbol{\mathscr{C}_g}=\{ (r(x)g_1(x),r(x)g_2(x),\dots,r(x)g_m(x)) \mid r(x)\in \hat{R}_{\alpha,\theta} \}\]
is called the \textit{$1$-generator skew quasi-twisted (SQT)
code} of length $mN$ and index $m$ with generator $\boldsymbol{g}$. Moreover, if $\alpha=1$ then $\boldsymbol{\mathscr{C}_g}$
is called the \textit{$1$-generator skew quasi-cyclic (SQC) code} of length $mN$ and index $m$ with generator
$\boldsymbol{g}$.
\end{defi}

First of all, observe that by using similar arguments as in \cite[Theorem 6]{B},
we can easily deduce the following two results.

\begin{thm}\label{thm 6}
Let $\mathcal{C}$ be a $1$-generator SQT code of length $mN$ and index $m$ over $\mathbb{F}_q$
generated by $\boldsymbol{g}=(g_1(x),g_2(x),\dots,g_m(x))\in \hat{R}_{\alpha,\theta}^m$ with
$N$ a multiple of the order of $\theta$ and $\alpha\in\mathbb{F}_q^\theta$.
Then $\{\boldsymbol{g},x\boldsymbol{g},...,x^{N-\deg g(x)}\boldsymbol{g} \}$ forms a $\mathbb{F}_q$-basis for $\mathcal{C}$, where
$g(x)$ is the great common left divisor $gcld(g_1(x), ... , g_m(x), x^N-\alpha)$.
\end{thm}

\begin{cor}\label{cor}\label{final cor}
Let $\mathcal{C}$ be a $1$-generator SQT code of length $mN$ and
index $m$ over $\mathbb{F}_q$ generated by
$\boldsymbol{g}=(g(x),g(x)p_1(x),\dots,g(x)p_{m-1}(x))\in
\hat{R}_{\alpha,\theta}^m$, where $g(x)$ is a monic divisor of
$x^N-\alpha$ with $N$ a multiple of the order of $\theta$ and
$\alpha\in\mathbb{F}_q^\theta$. Then $\mathcal{C}$ is an
$\mathbb{F}_q$-free code with rank $N-\deg g(x)$.
\end{cor}

Finally, from Theorem  \ref{GMaruta:3} we know that an $[N,N-k]_q$-code is a skew $(\alpha,\theta)$-cyclic code with generator polynomial
$g(x)$ if and only if
$\mathscr{C}$ is a linear code with a parity check matrix
$$[g]^N:=[{}^t \! P,{}^t \! \tau (P),{}^t \! \tau^2 (P),\dots,{}^t \! \tau^{N-1}(P)],$$ where
$P=(1,0,\dots,0)\in\F_q^k$, $\tau (\vec{v}):=\Theta (\vec{v})T_{g(x)}$ for every $\vec{v}\in\mathbb{F}^k_q$
and $T_{g(x)}$ is the companion matrix of $g(x)$.

Let $g:=x^k-\sum_{i=0}^{k-1} a_ix^i\in R$ be a right divisor of the polynomial $x^N-\alpha\in R$ for some $\alpha\in\mathbb{F}_q^*$
and denote by $\mathcal{T}:\P^{k-1}(\F_q)\to\P^{k-1}(\F_q)$ the map
defined by $\tau$, i.e. $\mathcal{T}([\vec{x}]_{\sim}):=[\tau(\vec{x})]_{\sim}$ for every $[\vec{x}]_{\sim}\in\P^{k-1}(\F_q)$.
We can say that $\mathcal{T}$
is defined by $g$. Note that the columns of  $[g]^N=:[a_0,a_1,...,a_{k-1}]^N$ can be
considered as points in $\P^{k-1}(\F_q)$ of an orbit of
$\mathcal{T}$. Conversely, we can obtain similarly as above a skew $(\alpha,\theta)$-cyclic
$[N,N-k]_q$-code from an orbit of length $N$ of the map $\mathcal{T}$.

So, consider $m$ orbits
$\mathcal{O}_1,\mathcal{O}_2,\dots,\mathcal{O}_m$ of $\mathcal{T}$
of length $n_i$ and starting points $P_i\in\mathcal{O}_i$ for
$i=1,...,m$, respectively. For simplicity, take $P_1=(1,0,\dots,0)\in\F_q^k$. Thus $n_1=N$ and
denote by $[g]^{N}+P_2^{n_2}+\cdots+P_m^{n_m}$ the matrix
\begin{multline*} [{}^t \! P_1,{}^t \! \tau (P_1),{}^t \! \tau^2 (P_1),\dots,{}^t \! \tau^{N-1}(P_1) , \
{}^t \! P_2,{}^t \! \tau (P_2),{}^t \! \tau^2 (P_2),\dots,{}^t \! \tau^{n_2-1}(P_2) , \
\cdots \\ \cdots ,\ {}^t \! P_m,{}^t \! \tau (P_m),{}^t \! \tau^2 (P_m),\dots,{}^t \! \tau^{n_m-1}(P_m)]\ . \end{multline*}
Note that the matrix $[g]^{N}+P_2^{n_2}+\cdots+P_m^{n_m}$ can be interpreted
as a generator matrix of a linear $[N+n_2+...+n_m,k]_q$-code.
This fact gives a simple method
to concatenating skew constacyclic codes to obtain new constructions of linear codes that meet the parameters of some best known linear codes
(BKLC) for small values of $N$ and $q$ (compare the examples in Table \ref{arco:generalizadoinferior1} with the corresponding codes in
\url{http://www.codetables.de/},\ \url{http://www.win.tue.nl/~aeb/}\ and the database of Magma in \cite{magma}).

\begin{ex}
Consider $\F_4=\{ 0,1,a,a^2\}$, where $a$ is a root of $x^2+x+1\in\F_2[x]$.
Then the polynomial $g(x)=x^6-(1+a^2x+x^2+a^2x^3+x^4+a^2x^5))\in\F_4[x;\theta]$ divides  $x^7-a$ ($N=7, \alpha=a$), where $\theta(z)=z^2$ for every $z\in\F_4$ ($p^t=2$).
Consider the points $P_1:=[1:0:0:0:0:0], P_2:=[a:a:0:a:1:1]\in\P^5(\F_4)$. Then, under $\mathcal{T}$, the orbit of $P_1$ is of length $7$
and that of $P_2$ is of length $14$.
So, we get a linear $[21,6,12]_4$-code whose generator matrix is
{\tiny
\begin{align*} \begin{pmatrix}
  1 &  0 &  0 &  0 &  0 &  0 &  1 &  a  & 1 & a^2 &  1 & a^2 & a^2 &  1 &  1 &  a & a^2 &  a & a^2 & a^2  & a \\
  0 &  1 &  0 &  0 &  0 &  0 & a^2  & a  & 0 & a^2  & 1 & a^2 &  0 &  1  &  a  & 0 &  1 & a^2 &  1 &  0 & a^2  \\
  0  & 0 &  1  & 0 &  0  & 0 &  1  & 0  & a & a^2 & a^2  & a &  1  & 1 &  0 &  1 & a^2 & a^2  & 1 &  a  & a \\
  0 &  0 &  0 &  1 &  0  & 0 & a^2 &  a & a^2 &  1 &  1 &  0  & 1 &  a &  a  & 1 & a^2 & a^2  & 0 & a^2 &  a \\
  0  & 0 &  0 &  0  & 1 &  0 &  1 &  1  & a &  1 &  0  & a & a^2 &  0 &  a  & 1 &  a &  0  & 1& a^2  & 0 \\
  0  & 0 &  0 &  0 &  0  & 1 & a^2  & 1 &  a &  1 &  a &  a &  1  & 1 & a^2  & a & a^2 &  a &  a & a^2 & a^2 \\
\end{pmatrix} \end{align*}}
and for simplicity of notation, the above matrix can be written simply as $[g]^7+P_2^{14}=
[1a^21a^21a^2]^7+[aa0a11]^{14}$ (see the first case for $q=4$ and $k=6$ in Table \ref{arco:generalizadoinferior1}).
\end{ex}

\begin{ex}
Consider $\F_8=\{ 0,1,w,w^2,\dots,w^6\}$.
Then the polynomial $g(x)=x^3-(1+w^6x+w^4x^2)\in\F_8[x;\theta]$ divides  $x^4-w$ ($N=4,\alpha=w$), where $\theta(z)=z^4$ for every $z\in\F_8$ ($p^t=2$).
Moreover,
we can see that the points $P_1:=[1:0:0], P_2:=[w^6:w:1], P_3:=[w^5:1:0], P_4:=[1:1:1]$ in $\P^3(\F_8)$ have orbits
of length $4,12,12,6$, respectively.
So, we get a linear $[34,3,28]_8$-code whose generator matrix can be written as
$$[g]^4+P_2^{12}+P_3^{12}+P_4^6=[1w^6w^4]^4+[w^6w1]^{12}+[w^510]^{12}+[111]^6$$
which reaches the best known distance for a linear $[34,3]_8$-code (see the first case for $q=8$ in Table \ref{arco:generalizadoinferior1}).
\end{ex}

\begin{ex}
Consider $\F_9=\{ 0,1,\beta,\beta^2,\dots,\beta^7\}$, where $\beta$ is a root of $x^2+2x+2\in\F_3[x]$.
We set that $g(x)=x^4+(1+\beta^7x+\beta^4x^2+\beta^3x^3)=x^4-(\beta^4+\beta^3x+x^2+\beta^7x^3) \in\F_9[x;\theta]$
divides $x^5-\beta$ ($N=5, \alpha=\beta$), with $\theta(z)=z^3$ for every $z\in\F_9$ ($p^t=3$).
Let $\mathcal{T}$ be the semi-linear map defined by $\theta$ and $g(x)$. Consider the following three points $P_1$, $P_2$ and $P_3$ of $\P^3(\F_9)$:
\[ P_1=[1:0:0:0],\ P_2=[1:\beta^5:\beta^7:1],\ P_3=[\beta^3:\beta^6:\beta^3:1]. \]
Then, under $\mathcal{T}$, the orbit of $P_1$ is of length $5$ and the orbits of $P_2$ and $P_3$ are both of length $10$, and the matrix
\[ [g]^5+P_2^{10}+P_3^{10}=[\beta^4\beta^31\beta^7]^5+[1\beta^5\beta^71]^{10}+[\beta^3\beta^6\beta^31]^{10} \]
generates a linear $[25,4,19]_9$-code (see the first case for $q=9$ in Table \ref{arco:generalizadoinferior1}).
\end{ex}

\medskip

\noindent On the other hand, when $\alpha =1$, $N$ is a multiple of the order of $\theta$, $n_1=n_2=...=n_m=N$ and $P_i\in(\F_q^{\theta})^k$ for $i=1,\dots, m$, the matrix
$[g^N]+P_2^N+\cdots+P_m^N$ becomes a generator matrix of a $1$-generator $SQC$ code of length $mN$ and index $m$
with generator polynomial $\boldsymbol{g}$ given
by the following result which generalizes \cite[Theorem 2.3]{M2} (see for instance Table \ref{Table 2} for some examples of $1$-generator $SQC$ codes for $q=4$).

\begin{thm}\label{GMaruta:QT}
With the same notation as above, if $P_i\in(\F_q^{\theta})^k$ for $i=1,\dots, m$, then
$[g^N]+P_2^N+\cdots+P_m^N$ generates a $1$-generator SQC $[mN,k]_q$-code of length $mN$ and index $m$ with generator
\[ \boldsymbol{g}=\left(h^{\ast}(x),h^{\ast}(x)b_2(x^{-1}),\dots,h^{\ast}(x)b_m(x^{-1})\right)\in \hat{R}^m,\]
where $$h^\ast(x):=\theta^{N-k}(\hbar_0)^{-1}\left[\sum_{i=1}^{N-k} \theta^i\left( \hbar_{N-k-i}\right)x^i+1\right]$$ with $\hbar(x)=\hbar_0+\cdots+\hbar_{N-k-1}x^{N-k-1}+x^{N-k}$ such that $x^N-\theta^{k-N}(\alpha)=g(x)\hbar(x)$ and $b_i(x)$ is the
polynomial given by $(1,x,\dots,x^{k-1})P_i$  for $2\leq i\leq m$.
\end{thm}

\begin{proof}
Define $H:=[g^N]+P_2^N+\cdots+P_m^N$ and use the same notation as above.
Note that
\begin{multline*} H=[{}^t \! P_1,{}^t \! \tau (P_1),{}^t \! \tau^2 (P_1),\dots,{}^t \! \tau^{N-1}(P_1) , \
{}^t \! P_2,{}^t \! \tau (P_2),{}^t \! \tau^2 (P_2),\dots,{}^t \! \tau^{N-1}(P_2) , \
\cdots \\ \cdots ,\ {}^t \! P_m,{}^t \! \tau (P_m),{}^t \! \tau^2 (P_m),\dots,{}^t \! \tau^{N-1}(P_m)]
\end{multline*}
with $P_1=\vec{e_1}:=(1,0,\dots,0)\in(\F_q^{\theta})^k$ and $\tau^N(P_i)=P_i$ for $i=1,...,m$. By putting
$$H_i:=[{}^t \! P_i,{}^t \! \tau (P_i),{}^t \! \tau^2 (P_i),\dots,{}^t \! \tau^{N-1}(P_i)]$$ for all
$i=1,\dots,m$, we simply have $H=[H_1|H_2|\cdots |H_m]$.

\bigskip

\noindent\textit{Claim $1$.} $\tau^h(P_i)=P_i(\tau)(\tau^h(P_1))$
\qquad $\forall h=0,\dots,N-1 \ .$

\medskip

\noindent Note that $\tau^{j-1}(P_1)=\vec{e_{j}}\in\F_q^k$ for all
$j=1,\dots, k$, where $\vec{e_j}$ is the $j$-th canonical vector
of $\F_q^k$. Therefore, by putting
$P_i:=(\lambda_{i,0},...,\lambda_{i,k-1})\in(\F_q^{\theta})^k$ for
$i=1,...,m$, we have
\[P_i=\sum_{j=0}^{k-1}\lambda_{ij} \tau^j(P_1)=\left(\sum_{j=0}^{k-1}\lambda_{ij} \tau^j\right)(P_1)
=:P_i(\tau)(P_1), \] where $P_i(z):=\sum_{j=0}^{k-1}\lambda_{ij}
z^j\in\mathbb{F}_q^{\theta}[z]$. Thus
$$\tau^h(P_i)=\sum_{j=0}^{k-1}\lambda_{ij}
\tau^{h+j}(P_1)=P_i(\tau)(\tau^h(P_1)) \qquad \forall
h=0,\dots,N-1 \ . \qquad\qquad Q.E.D. $$

\bigskip

Define $\tau^s\cdot [{}^t \! Q_0,{}^t \! Q_2,\dots,{}^t \! Q_{N-1})]:=[{}^t \! \tau^s(Q_0),{}^t \! \tau^s(Q_2),\dots,{}^t \! \tau^s(Q_{N-1})]$
for any $s\in\mathbb{Z}_{\geq 0}$.

\bigskip

\noindent\textit{Claim $2$.} $H_i=P_i(\tau)\cdot H_1$ \qquad
$\forall i=1,...,m$.

\medskip

\noindent For $i=1,\dots,m$, by \textit{Claim $1$} we see that
\begin{align*}
H_i&=[{}^t \! P_i,{}^t \! \tau (P_i),{}^t \! \tau^2 (P_i),\dots,{}^t \! \tau^{N-1}(P_i)]\\
&=[{}^t \! P_i(\tau)(P_1),{}^t \! P_i(\tau)(\tau (P_1)),{}^t \! P_i(\tau)(\tau^2 (P_1)),\dots,{}^t \! P_i(\tau)(\tau^{N-1}(P_1))] \\
&=:P_i(\tau)\cdot [{}^t \! P_1,{}^t \! \tau (P_1),{}^t \! \tau^2
(P_1),\dots,{}^t \! \tau^{N-1}(P_1)]=P_i(\tau)\cdot H_1\ . \qquad\qquad\qquad Q.E.D.
\end{align*}

\bigskip

Let $H_1^\ast$ be a parity check matrix of the code $\mathscr{C}$ such that $\pi(\mathscr{C})=Rg(x)/R(x^N-1)$,
written as
$$ H_1^\ast\equiv\begin{pmatrix} h^\ast (x)\\ xh^\ast (x)
\\ \vdots \\ x^{k-1}h^\ast(x)
\end{pmatrix} =[ J \mid
\widehat{H}_1^\ast ] $$ with $\det(J)\not= 0$, where from Proposition \ref{asterisco} it follows that
$$h^\ast(x):=\theta^{N-k}(\hbar_0)^{-1}\left[\sum_{i=1}^{N-k} \theta^i\left( \hbar_{N-k-i}\right)x^i+1\right]$$
with $\hbar(x)=\hbar_0+\cdots+\hbar_{N-k-1}x^{n-k-1}+x^{N-k}$ such that $x^N-\theta^{k-N}(\alpha)=g(x)\hbar(x).$

\bigskip

\noindent\textit{Claim $3$.} $H_i=J^{-1}H_i^*$ \quad $\forall
i=1,...,m$, \ where $H_i^\ast:=\begin{pmatrix}
h^\ast (x)P_i(x^{-1})\\ x (h^\ast (x)P_i(x^{-1})) \\
\vdots \\ x^{k-1} (h^\ast(x)P_i(x^{-1}))\end{pmatrix}$.

\medskip

\noindent By hypothesis, $H_1$ is the parity check matrix of
$\mathscr{C}$ in the standard form which can be written as
$H_1=[I_k\mid \widehat{H}_1]$, where $I_k$ is the $k\times k$
identity matrix. In particular, this implies that
$H_1=J^{-1}H_1^\ast$.
Now, let $A$ be a matrix with $N$ columns and for any
$h\in\mathbb{Z}_{\geq 0}$ define the bilinear map
$$\diamond :\ \mathbb{F}_q[x^{-1}]\times \mathrm{Mat}(k,N;\mathbb{F}_q)\to
\mathrm{Mat}(k,N;\mathbb{F}_q)$$
$\forall \lambda,\mu\in\mathbb{F}_q,\ \forall k,h\in\mathbb{Z}_{\geq 0}$ and $\forall A,B\in \mathrm{Mat}(k,N;\mathbb{F}_q)$ as follows:
\[ x^{-h}\diamond A := A\left( \begin{array}{ c | c }
\vec{0} &  1 \\  \hline
I_N & {}^t \! \ \vec{0} \\
 \end{array}\right)^h = A \left( \begin{array}{ c  c c | c }
0 &  \cdots &0 &  1 \\  \hline
1 &  \cdots &0 & 0 \\
\vdots & \ddots & \vdots & \vdots \\
0 &\cdots & 1  & 0 \\
 \end{array}\right)^h
 \]
  $$(\lambda x^{-h}+\mu x^{-k})\diamond A= \lambda (x^{-h}\diamond A)+\mu (x^{-k}\diamond A),$$
 $$x^{-h}\diamond (\lambda A+\mu B)=\lambda (x^{-h}\diamond A)+\mu (x^{-h}\diamond A)\ .$$
\noindent Moreover, observe that $\lambda x^{-h}\diamond (CD)=C
(\lambda x^{-h}\diamond D)$ for every $\lambda\in\mathbb{F}_q$, $C\in\mathrm{Mat}(k,k;\mathbb{F}_q)$ and $D\in\mathrm{Mat}(k,N;\mathbb{F}_q)$. Furthermore, if
$A = \begin{pmatrix} \vec{a}_0
\\ \vdots \\ \vec{a}_{k-1}
\end{pmatrix}
\equiv \begin{pmatrix} a_0(x)
\\ \vdots \\ a_{k-1}(x)
\end{pmatrix}$, where $\pi: \mathbb{F}_q^N\to R/R(x^N-1)$ and $\pi (\vec{a}_j)=a_j(x)$ for $j=0,...,k-1$, then we have
$x^{-1}\diamond A \equiv \begin{pmatrix} a_0(x)x^{-1}
\\ \vdots \\ a_{k-1}(x)x^{-1}
\end{pmatrix}$, where $p(x)x^{-1}=p_1+p_2x+...+p_{N-1}x^{N-2}+p_0x^{N-1}$ if $p(x)=p_0+p_1x+...+p_{N-1}x^{N-1}$, since $x^{N}=1$ and $x^{-1}=x^{N-1}$.
Hence we get $(\lambda x^{-h}+\mu x^{-k})\diamond A \equiv \begin{pmatrix} \lambda a_0(x)x^{-h}
\\ \vdots \\ \lambda a_{k-1}(x)x^{-h}
\end{pmatrix}+\begin{pmatrix} \mu a_0(x)x^{-h}
\\ \vdots \\ \mu a_{k-1}(x)x^{-h}
\end{pmatrix}$ and
\begin{align*}
x^{-1}\diamond H_1 &= x^{-1}\diamond [{}^t \! P_1,{}^t \! \tau (P_1),{}^t \! \tau^2 (P_1),\dots,{}^t \! \tau^{N-1}(P_1)]\\ &
=[{}^t \! \tau (P_1),{}^t \! \tau^2 (P_1),\dots,{}^t \! \tau^{N-1}(P_1),{}^t \! P_1]\\&
=[{}^t \! \tau (P_1),{}^t \! \tau^2 (P_1),\dots,{}^t \! \tau^{N-1}(P_1),{}^t \! \tau^{N}(P_1)]\\&=\tau\cdot H_1\\
x^{-2}\diamond H_1&=[{}^t \! \tau^2 (P_1),\dots,{}^t \! \tau^{N-1}(P_1),{}^t \! P_1,{}^t \! \tau(P_1)]=\tau^2\cdot H_1 \\
&\ \ \vdots\\
x^{-h}\diamond  H_1&=\tau^h\cdot H_1\ .
\end{align*}
Then we deduce that $q(\tau)\cdot H_1=q(x^{-1})\diamond  H_1$ for
any polynomial $q(t)\in\mathbb{F}_q[t]$, where
$(x^{-1})^h:=x^{-h}$. Keeping in mind that
$P_i(z)\in\mathbb{F}_q^{\theta}[z]$, by \textit{Claim $2$} we have
for $i=1,...,m$ 
\begin{align*}
H_i&= P_i(\tau) \cdot H_1=
P_i(x^{-1})\diamond (J^{-1} H_1^\ast)
=J^{-1} (P_i(x^{-1})\diamond H_1^\ast)=\\
&=J^{-1} \left(P_i(x^{-1})\diamond  \begin{pmatrix}
h^\ast (x)\\ xh^\ast (x) \\ \vdots \\ x^{k-1}h^\ast(x)
\end{pmatrix}\right)=J^{-1} \begin{pmatrix}
(h^\ast (x))P_i(x^{-1})\\ (xh^\ast (x))P_i(x^{-1}) \\ \vdots \\ (x^{k-1}h^\ast(x))P_i(x^{-1})
\\
\end{pmatrix}= \\
& =J^{-1} \begin{pmatrix}
h^\ast (x)P_i(x^{-1})\\ x (h^\ast (x)P_i(x^{-1})) \\
\vdots \\ x^{k-1} (h^\ast(x)P_i(x^{-1}))\end{pmatrix} =:J^{-1}
H^\ast_i \ . \qquad \qquad \qquad \qquad Q.E.D. 
\end{align*}

\noindent So, by \textit{Claim $3$} we have
$$H=[H_1|H_2|\cdots |H_m]=[J^{-1}H_1^\ast|J^{-1}H_2^\ast|\cdots |J^{-1}H_m^\ast]
=J^{-1}[H_1^\ast | H_2^\ast | \cdots | H_m^\ast ],
$$
where $H_i^\ast:=\begin{pmatrix}
h^\ast (x)P_i(x^{-1})\\ x (h^\ast (x)P_i(x^{-1})) \\
\vdots \\ x^{k-1} (h^\ast(x)P_i(x^{-1}))\end{pmatrix}$ for
$i=1,...,m$.
Finally, by Corollary \ref{final cor} we can conclude that
$H=[H_1|H_2|\cdots |H_m]$ is a generator matrix of a $1$-generator
SQC $[mN,k]_q$-code of length $mN$ and index $m$ with generator
polynomial
\[ \boldsymbol{g}=\left(h^{\ast}(x),h^\ast (x)P_2(x^{-1}),\dots,h^\ast (x)P_m(x^{-1})\right)\in \hat{R}^m,\]
where $h^\ast(x)$ is given by Proposition \ref{asterisco}.
\end{proof}

{\tiny
\begin{table}[!h]
\begin{center}\renewcommand{\arraystretch}{1.6}
\setlength{\tabcolsep}{8pt}
\begin{tabular}{ c p{8cm} c c c }\hline
$[n,k,d]_q$ & Generator Matrix  & $N$ & $\alpha$ & $p^t$ \\ \hline
$[25,4,17]_4$ & $[1a^21a^2]^5+[a^2a^2a^21]^{10}+[aa^210]^{10}$ & 5 & $a$ & 2\\
$[35,4,24]_4$ & $[1a1a]^5+[0a^2a1]^{10}+[10a^21]^{10}+[1a^201]^{10}$ & 5 & $a^2$ & 2 \\
$[40,4,28]_4$ & $[1a^21a^2]^5+[0a^2a1]^{10}+[a^2a^201]^{10}+[1a^201]^{10}+[1010]^5$ & 5 & $a$ & 2 \\
$[45,4,32]_4$ & $[1a1a]^5+[a100]^{10}+[1a^211]^{10}+[0a01]^{10}+[a^2a^2a1]^{10}$ & 5 & $a^2$ & 2 \\
$[50,4,36]_4$ & $[1a1a]^5+[1aa^21]^{10}+[1101]^{10}+[aaa^21]^{10}+[a^2a^2a1]^{10}+[a^2100]^5$ & 5 & $a^2$ & 2 \\
$[60,4,44]_4$ & $[1a^21a^2]^5+[1110]^{10}+[0aa^21]^{10}+[00a^21]^{10}+[a^2111]^{10} +[a^2a^211]^{10}+[a0a1]^5$ & 5 & $a$ & 2\\
$[65,4,48]_4$ & $[1a1a]^5+[aa01]^{10}+[1011]^{10}+[a1a^21]^{10} +[1a^210]^{10}+[0a11]^{10}+[a^2100]^5+[a^20a^21]^{5}$ & 5 & $a^2$ & 2 \\
$[85,4,64]_4$  &$[1a1a]^5+[a^2a11]^{10}+[0011]^{10}+[a^2a^201]^{10}+[1a^201]^{10}+[0a^2a^21]^{10}+[a^21a1]^{10}+[1a^2a^21]^{10}+[1010]^{5}+[0a10]^5$  & 5  & $a^2$ & 2   \\
$[12, 5, 6]_4$ & $[a^21a^21a^2]^{6}+[1a^2a^2aa^2]^{6}$    &   6  &  1  &  2 \\
$[18, 5, 10]_4$ & $[a^21a^21a^2]^{6}+[01a^201]^{6}+[a^20a^2aa]^{6}$ & 6  &  1 &  2 \\
$[20, 5, 12]_4$ & $[aa110]^{10}+[a^20a^2aa]^{10}$  &  10   &  1  &  2  \\
$[30, 5,20]_4$   &   $[11a^2a1]^{10}+[a^2aaa^21]^{10}+[01a^201]^{10}$   &   10  &  1  &  2 \\
$[32, 5, 21]_4$  &  $[a^2aa^2a^20]^{8}+[aa1a^20]^{8}+[0a^20a^2a]^{8}+[1001a^2]^{8}$    &  8   &  1  &  2 \\
$[36, 5, 24]_4$  &  $[a^2aa^20a^2]^{12}+[a^2aaa^21]^{12}+[01a^201]^{12}$    &  12   &  1  &  2 \\
$[48, 5, 33]_4$  &  $[1101a^2]^{12}+[a^20a^2aa]^{12}+[0a^20a^2a]^{12}+[a^2a^2a^21a]^{12}$    &  12   &  1  &  2 \\
$[120, 5, 88]_4$  &  $[a^21a1a^2]^{30}+[a^2aaa^21]^{30}+[a^20a^2aa]^{30}+[0a^200a]^{30}$    &  30   &  1  &  2 \\
$[248, 5, 184]_4$  &  $[aa^2a^2aa]^{62}+[a^2aaa^21]^{62}+[aa1a^20]^{62}+[0a^200a]^{62}$    &  62   &  1  &  2 \\
$[21,6,12]_4$ & $[1a^21a^21a^2]^7+[aa0a11]^{14}$ &  $7$ & $a$ &  2 \\
$[49,6,32]_4$ & $[1a^21a^21a^2]^7+[aa^2a^2101]^{14}+[a^21aaa1]^{14}+[aa^21a^210]^{14}$ & 7 & $a$ & 2 \\[5pt]
\hline 
$[34,3,28]_8$ & $[1w^6w^4]^4+[w^6w1]^{12}+[w^510]^{12}+[111]^6$ & 4 & $w$ & 2 \\
$[50,4,41]_8$ & $[w^51w^4w^5]^5+[w^4w^2w^21]^{15}+[w^4w^510]^{15}+[w^5w^2w1]^{15}$ & 5 & $w$ & 2 \\
$[65,4,56]_8$ & $[w^51w^4w^5]^5+[w^4w^4w^31]^{15}+[w^6w^210]^{15}+[w^2110]^{15}+[1w^4w^21]^{15}$ & 5 & $w$ & 2 \\[5pt]
\hline
$[25,4,19]_9$ & $[\beta^4\beta^31\beta^7]^{5}+[1\beta^5\beta^71]^{10}+[\beta^3\beta^6\beta^31]^{10}$ & 5 & $\beta$ & 3 \\
$[42,4,34]_9$ & $[10\beta^20]^{6}+[\beta^71\beta 1]^{12}+[\beta^5\beta^7\beta 1]^{12}+[0\beta^4\beta^21]^{12}$ & 6 & $\beta^2$ & 3 \\[5pt]
\hline \\
\end{tabular}
\caption{New constructions of some linear codes with BKLC parameters.}
\label{arco:generalizadoinferior1}
\end{center}
\end{table}}

{\tiny
\begin{table}[!h]
\begin{center}\renewcommand{\arraystretch}{1.6}
\setlength{\tabcolsep}{8pt}
\begin{tabular}{ c p{8cm} c c c }\hline
$[n,k,d]_q$ & Generator Matrix  & $N$ & $\alpha$ & $p^t$ \\ \hline
$[12,5,6]_4$ & $[a^21a^21a^2]^{6}+[01111]^{6}$ & 6 & $1$ & 2\\
$[20,5,12]_4$ & $[aa110]^{10}+[01011]^{10}$ & 10 & $1$ & 2\\
$[30,5,20]_4$ & $[aa110]^{10}+[11011]^{10}+[10011]^{10}$ & 10 & $1$ & 2\\
$[36,5,24]_4$ & $[1101a^2]^{12}+[01011]^{12}+[10001]^{12}$ & 12 & $1$ & 2\\
$[48,5,33]_4$ & $[10a^2a0]^{12}+[11010]^{12}+[00111]^{12}+[10111]^{12}$ & 12 & $1$ & 2\\
$[90,5,64]_4$ & $[a^210a1]^{30}+[01011]^{30}+[00111]^{30}$ & 30 & $1$ & 2\\
$[120,5,88]_4$ & $[a^20aa^2a]^{30}+[00011]^{30}+[10111]^{30}+[10010]^{30}$ & 30 & $1$ & 2\\
$[248,5,184]_4$ & $[10a^210]^{62}+[00011]^{62}+[11010]^{62}+[11101]^{62}$ & 62 & $1$ & 2\\[2pt]
\hline \\
\end{tabular}
\caption{Some $1$-generator $SQC$ codes for $q=4$ and $k=5$ with BKLC parameters.}
\label{Table 2}
\end{center}
\end{table}}

\subsection{MDS skew $(\alpha,\theta)$-cyclic codes}\label{MDS}

First of all, let us observe that in the commutative case, from
\cite[Theorem 6]{M1} it is known that there exists a MDS
$\alpha$-constacyclic $[n,k]$-code over $\F_q$ with $(n,q)\not=1$ and $2\leq
k\leq n-2$, if and only if $n=p$.

Let us show here that in the non-commutative case, there exist MDS 
skew constacyclic $[n,k]_q$-codes with $(n,q)\not=1$ and $2\leq k\leq n-2$ also when $n\not= p$.

\begin{ex}
Over $\F_4=\F_2[\alpha]$ with $\alpha^2+\alpha+1=0$ and $\theta (z)=z^2$, consider the skew
$(\alpha,\theta)$-cyclic $[4,2]_4$-code $\mathscr{C}$ generated by
$g(x)=\alpha x^2+\alpha^2 x+\alpha^2$. Then $\mathscr{C}$ is a MDS code with parity check matrix
$$\begin{pmatrix}
\alpha^2 & \alpha^2 & \alpha & 0 \\
0 & \alpha & \alpha & \alpha^2 \\
\end{pmatrix}.$$
\end{ex}

\begin{ex}
Over $\F_9=\F_3[\omega]$ with $\omega^3+\omega+2=0$ and $\theta (z)=z^2$, consider the skew
$(1,\theta)$-cyclic $[6,4]_9$-code $\mathscr{C}$ generated by
$g(x)=\omega^5 x^2+\omega^7 x+\omega^7$. Then $\mathscr{C}$ is a MDS code with parity check matrix
{\footnotesize $$\begin{pmatrix}
\omega^7 & \omega^7 & \omega^5 & 0 & 0 & 0 \\
0 & \omega^5 & \omega^5 & \omega^7 & 0 & 0 \\
0 & 0 & \omega^7 & \omega^7 & \omega^5 & 0  \\
0 & 0 & 0 & \omega^5 & \omega^5 & \omega^7  \\
\end{pmatrix}.$$ }
\end{ex}

\smallskip

By Theorem \ref{GMaruta:2}, we recall that the parity check matrix of a MDS skew $(\alpha,\theta)$-cyclic $[n,n-k]_q$-code has the form 
$[{}^t P, 
{}^t \tau(P), {}^t \tau
^2(P),\dots,{}^t \tau^{n-1}(P)],$ where $P\in\mathbb{F}^k_q$, $\tau (\vec{v}):=\Theta (\vec{v})T$ for every $\vec{v}\in\mathbb{F}^k_q$ with $T\in GL(k,q)$ and $\tau^n(P)=\alpha P$ with $n\leq \text{ord}(\tau)$.

Therefore, under certain conditions on $ q $, $ k $ and $ n $, the
existence of MDS skew $(\alpha,\theta)$-cyclic $[n, k]_q $-codes
is strictly related to some algebraic conditions, as it is shown in the following results.

\begin{prop}\label{prop MDS k>3}
Assume that $q>2$, $k\geq 4$ and
\[ \begin{cases}
q+k-\dfrac{\phantom{|}\sqrt{q}+5\phantom{|}}{4}< n &, \text{ for } q \text{ odd}\\
q+k-\dfrac{2\sqrt{q}+7}{4}<n &, \text{ for } q \text{ even}\ .
\end{cases} \]
If there exists a MDS skew $(\alpha,\theta)$-cyclic $[n,n-k]_q$-code, then $n\leq q+1$.
\end{prop}

\begin{proof}
Let $\mathscr{C}\subseteq\F_q^n$ be a MDS skew $(\alpha,\theta)$-cyclic $[n,n-k]_q$-code. By Theorem  \ref{GMaruta:2}, a parity check matrix of $\mathscr{C}$ has the form $[{}^t P, 
{}^t \tau(P), {}^t \tau
^2(P),\dots,{}^t \tau^{n-1}(P)],$ where $P\in\mathbb{F}^k_q$ and $\tau=T\circ\Theta$ for some $T\in GL(k,q)$ with
 $\tau (\vec{v}):=\Theta (\vec{v})T$ for every $\vec{v}\in\mathbb{F}^k_q$ and such that $\tau^n(P)=\alpha P$.
Note that the set $\mathcal{K}:=\{ [\tau^i(P)] \colon i=0,\dots,n-1 \}\subseteq \mathbb{P}^{k-1}(\F_q)$,
defines an $n$-arc in the finite projective space $\mathbb{P}^{k-1}(\F_q)$. Thus, by \cite{BBT1, BBT2, KM, ST, T} and the hypothesis
\[
\begin{cases}
q-\frac{1}{4}\sqrt{q}+(k-1)-\frac{1}{4}<n&, q\text{ odd}\\
q-\frac{1}{2}\sqrt{q}+(k-1)-\frac{3}{4}<n&, q\text{ even}\ ,
\end{cases}
\]
we deduce that $\mathcal{K}$ lies on a unique rational normal curve. Hence $n=|\mathcal{K}|\leq q+1$.
\end{proof}

Let us note here that 
for some $a\in\mathbb{F}_q^*$, a skew monic polynomial $F(x)\in
\mathbb{F}_q[x;\theta]$ could not divide $x^m-a$ on the right for
every $m\in\mathbb{Z}_{\geq 0}$. Nevertheless, in \cite[$\S
3.2.1$]{TT1} it was shown that any skew polynomial $F(x)$ as above
with a regular constant term always divides the polynomial $x^m-1$
on the left for some $m\in\mathbb{Z}_{>0}$. Observe that the proof
of this fact works well also for right divisions with only slight
modifications.

So, in line with the classical definition of subexponent of a
polynomial in $\mathbb{F}_q[x]$ (see, e.g., \cite[pp. 6--7]{H}),
let us give here a similar definition for the non-commutative ring
$\mathbb{F}_q[x;\theta]$ (see also \cite[Definitions 2.1(a) and 3.1]{CL}). 

\begin{defi}
Let $F(x)$ be a skew monic polynomial in $\mathbb{F}_q[x;\theta]$.
If $F(0)\neq 0$, then $F(x)$ has \textit{right exponent} $e(F)$
if $e(F)$ is the smallest positive integer such that $F(x)$ is a
right divisor of $x^{e(F)}-a$ for some $a\in\mathbb{F}_q^*$. Similarly, 
one can define the notion of \textit{left exponent} of $F(x)$.
\end{defi}

The following result is an algebraic characterization
of the existence of some MDS skew $(\alpha,\theta)$-cyclic codes with certain parameters.

\begin{thm}\label{teo:cuaderno2}
Assume that $q>2$, $k\geq 4$ and
\[ \begin{cases}
q+k-\dfrac{\phantom{|}\sqrt{q}+5\phantom{|}}{4}< n \leq q+1 &, \text{ for } q \text{ odd}\\
q+k-\dfrac{2\sqrt{q}+7}{4}< n \leq q+1 &, \text{ for } q \text{ even}.
\end{cases} \]
Then, there exists a {\em MDS} skew $(\alpha,\theta)$-cyclic
$[n,n-k]_q$-code if and only if there exists a polynomial $x^2+a
x+b\in\F_q[x;\theta]$ with $b\neq 0$ such that $e(x^2+ax+b)=n$.
\end{thm}

\begin{proof}
Let $\mathscr{C}\subseteq\F_q^n$ be a MDS skew $(\alpha,\theta)$-cyclic $[n,n-k]_q$-code. By using the same notation as in Proposition \ref{prop MDS k>3},
through a projectivity given by $A\in GL(k,q)$, we can send the points of $\mathcal{K}$ onto the canonical rational normal curve which is the image of the Venorese map $\nu_k\colon \mathbb{P}^1(\F_q) \longrightarrow \mathbb{P}^{k-1}(\F_q)$ given by $$\nu_k[x_0:x_1]=[x_0^{k-1}:x_0^{k-2}x_1:\dots:x_0x^{k-2}:x^{k-1}]\ .$$
\noindent By taking $\tau':=A\circ \tau\circ  A^{-1}$ given by $\tau'(Q):=A\circ \tau\circ  A^{-1}(Q)=A\circ \tau (QA^{-1})=\tau (QA^{-1})A$ for every $Q\in\mathbb{F}_q^k$, we can define the set
\[ \mathcal{K}':=\{ [\tau^i(P) A] \colon i=0,\dots,n-1 \}=\{ [(\tau')^i(PA)] \colon i=0,\dots,n-1 \}\subseteq \mathbb{P}^{k-1}(\F_q), \]
where $\tau'=M'\circ\Theta$ for some $M'\in GL(k,q)$.
Now, consider the following commutative diagram
$$\begin{CD}
\mathbb{P}^{k-1}(\F_q) @>{\tau':=\ M'\circ\Theta}>> \mathbb{P}^{k-1}(\F_q) \\
@ A{\nu_k}AA  @ AA{\nu_k }A \\
\mathbb{P}^1(\F_q) @>>{\ \  ~~\quad~ ~ ~ \sigma:=M\circ\Theta \ \ ~ ~\quad~ ~ ~ }> \mathbb{P}^1(\F_q)
\end{CD}\quad ,$$
where $M\in GL(2,q)$ is such that  $\nu_k\circ\sigma= \tau'\circ\nu_k$. In this way, we can identify the elements of  $\mathcal{K}'$ 
with the elements of $ \left\{ [(\sigma)^i(\vec{p})]=[(M\circ\Theta)^i(\vec{p})]\colon i=0,\dots,n-1 \right\}\subseteq \P^1(\F_q)$,
where $[\vec{p}]:=\nu_k^{-1}([PA])$. 
Consider a projectivity of $\mathbb{P}^1(\mathbb{F}_q)$ given by $B\in GL(2,q)$ such that $(1,0)B=\vec{p}$. Let
$M''=\begin{pmatrix}
\alpha & \beta\\
\gamma & \delta
\end{pmatrix}$
be a $2\times 2$ matrix with $\beta\neq 0$ such that the following diagram commutes,
$$\begin{CD}
\mathbb{P}^1(\F_q) @>{\sigma:=M\circ\Theta}>> \mathbb{P}^1(\F_q) \\
@ A{B}AA  @ AA{B}A \\
\mathbb{P}^1(\F_q) @>>{\sigma'':=M''\circ\Theta}> \mathbb{P}^1(\F_q)
\end{CD} \quad , $$
where $[(1,0)]B=[\vec{p}]$. Define
$C:=\begin{pmatrix}
1 & 0\\
\alpha & \beta
\end{pmatrix}$ and observe that
$C^{-1}=\beta^{-1}\begin{pmatrix}
\beta & 0\\
-\alpha & 1
\end{pmatrix}$. 

Then
$$C^{-1}\circ (M''\circ\Theta)\circ C=
\begin{pmatrix}
1 & 0\\
\theta(\alpha) & \theta(\beta)
\end{pmatrix}
\begin{pmatrix}
\alpha & \beta\\
\gamma & \delta
\end{pmatrix}
\begin{pmatrix}
1 & 0\\
-\alpha\beta^{-1} & \beta^{-1}
\end{pmatrix}\circ\Theta=\begin{pmatrix}
0 & 1\\
-b & -a
\end{pmatrix}\circ\Theta $$
for some $a,b\in\mathbb{F}_q$ with $b\neq 0$. This gives the further commutative diagram
$$\begin{CD}
\mathbb{P}^1(\F_q) @>{\sigma'':=M''\circ\Theta}>> \mathbb{P}^1(\F_q) \\
@ A{C}AA  @ AA{C}A \\
\mathbb{P}^1(\F_q) @>>{A\circ\Theta}> \mathbb{P}^1(\F_q)
\end{CD} \quad , $$
where $A:=\begin{pmatrix}
0 & 1\\
-b & -a
\end{pmatrix}$ and $[(1,0)]C=[(1,0)]$.
Combining the previous commutative diagrams, we obtain the following commutative diagram
$$\begin{CD}
\mathbb{P}^{k-1}(\F_q) @>{\tau':=\ M'\circ\Theta}>> \mathbb{P}^{k-1}(\F_q) \\
@ A{\nu_k\circ B\circ C}AA  @ AA{\nu_k\circ B\circ C}A \\
\mathbb{P}^1(\F_q) @>>{\ \  ~~\quad~ ~ ~ A\circ\Theta \ \ ~ ~\quad~ ~ ~ }> \mathbb{P}^1(\F_q)
\end{CD}\quad ,$$
such that the elements of  $\mathcal{K}'$ can be identified with
the elements of $$\left\{ [(A\circ\Theta)^i(1,0)]\colon
i=0,\dots,n-1 \right\}\subseteq \P^1(\F_q)\ .$$ Let $f(x):=x^2+ax
+b$ be the characteristic polynomial of $A$ and define the
bijective map
$$\pi\colon \F_q^2\longrightarrow \F_q[x;\theta]/(f(x)) =:R$$ defined by
$\pi(\vec{v}):=v_0+v_1x$, where $\vec{v}:=(v_0,v_1)$. 

Observe that $\pi\left((A\circ\Theta)^j(\vec{v})\right)=x^j\cdot \pi(\vec{v})=x^j(v_0+v_1x)$ for any $j\in\mathbb{Z}_{\geq 0}$ and that $\pi$ induces the bijective map
$$\widehat{\pi} \colon \left( \F_q^2\setminus\{ (0,0)\} \right)/\ \F_q^\ast \longrightarrow \left(R\setminus\{ 0 \} \right)\ /\ \F_q^\ast $$ given by $\widehat{\pi}[(v_0,v_1)]:=[\pi(v_0,v_1)].$
So, $\widehat{\pi}[(A\circ\Theta)^j(\vec{v}) ]=[x^j( v_0+v_1x)]$ for any $j\in\mathbb{Z}_{\geq 0}$.
Moreover, since $\left( \F_q^2\setminus\{ (0,0)\} \right)/\ \F_q^\ast \cong \P^1(\F_q)$, for all $i=1,\dots,n-1$, we have
\begin{align*}
[(1,0)]\not=[(A\circ\Theta)^i(1,0)]\Longrightarrow \widehat{\pi}[(1,0)] \not= \widehat{\pi}[(A\circ\Theta)^i(1,0)]
& \Longrightarrow[1]\not= [x^i]  \\
& \Longrightarrow
\lambda \not=x^i  \tag{ $\forall \lambda\in\F_q^\ast$} \\
 & \Longrightarrow
 x^i  - \lambda \not\equiv 0 \text{ mod }f(x) \tag{ $\forall \lambda\in\F_q^\ast$}
\end{align*}
i.e., there exists $f(x)=x^2+ax+b\in\F_q[x;\theta]$ with $b\neq 0$
such that $x^i-\lambda\not\equiv \text{mod }f(x)$ for all
$\lambda\in\F_q^\ast$ and $i=1,\dots,n-1$. Since $[1]=[x^n]$, we
conclude that there exists $x^2+ax+b\in\F_q[x;\theta]$ with $b\neq
0$ such that $e(x^2+ax+b)=n$.

\medskip

\noindent Looking closely the above proof, by
construction the converse of the statement becomes at this point
easy to prove.
\end{proof}

\medskip

Finally, observe that the above result provides a method to construct via Veronese embeddings some MDS skew $(\alpha,\theta)$-cyclic
$[n,n-k]_q$-code.

\begin{cor}\label{existence cor}
If a polynomial $x^2+ax+b\in\F_q[x;\theta]$ with $b\neq 0$ has right exponent $e(x^2+ax+b)=n\leq q+1$, then
for any integer $h$ such that $1\leq h\leq n-1$
there exists a {\em MDS} skew $(\alpha,\theta)$-cyclic
$[n,h]_q$-code for some $\alpha\in\mathbb{F}_q^*$.
\end{cor}

\begin{proof} Let $k$ be an integer such that $2\leq k\leq n-1$ and define $A:=\begin{pmatrix}
0 & 1\\
-b & -a
\end{pmatrix} $. Note that there exists a matrix $M\in GL(k,q)$ such that the following diagram commutes:
$$ \begin{CD}
\mathbb{P}^{k-1}(\F_q) @>{M\circ\Theta}>> \mathbb{P}^{k-1}(\F_q) \\
@ A{\nu_k}AA  @ AA{\nu_k}A \\
\mathbb{P}^1(\F_q) @>>{A\circ\Theta}> \mathbb{P}^1(\F_q)
\end{CD} \quad , $$
where $\nu_k[x_0:x_1]=[x_0^{k-1}:x_0^{k-2}x_1:\dots:x_0x^{k-2}:x^{k-1}]$. Thus the set 
$$\left\{ [(A\circ\Theta)^i(1,0)]\colon i=0,\dots,n-1 \right\}\subseteq \P^1(\F_q) $$
provides a set $ \left\{ [(M\circ\Theta)^i(1,0,\dots , 0)]\colon i=0,\dots,n-1 \right\}\subseteq \P^{k-1}(\F_q)$ of points lying on a rational normal curve in $\P^{k-1}(\F_q)$.
This gives a matrix of type
$$[{}^t P, {}^t \tau(P), {}^t \tau^2(P),\dots,{}^t \tau^{n-1}(P)]\ ,$$ where $P=(1,0,\dots , 0)\in\mathbb{F}^k_q$ and $\tau := M\circ\Theta$ 
is such that $\tau^n(P)=\beta P$ for some $\beta\in\mathbb{F}_q^*$. We conclude by Theorems \ref{GMaruta:1e}, \ref{GMaruta:2} and Corollary \ref{cor gen}.
\end{proof}

When $\theta$ is the Frobenius automorphism of $\mathbb{F}_q$, in \cite[Theorem 3.1]{CL} the authors show that the right and left exponents of a skew polynomial $g(x)$ in $\mathbb{F}_q[x;\theta]$ are equal and they provide an algebraic method to compute the right exponent of $g(x)$. On the other hand, the following MAGMA program defines 
a command EXP$(p^t,[b,a,1])$
which determines the right exponent of any skew polynomial
$x^2+ax+b\in\mathbb{F}_q[x;\theta]$ with $b\neq 0$,
when $\theta(z)=z^{p^t}$:

\begin{prog}\label{Prog-2}~
\begin{verbatim}
c:= ... ; F<w>:=GF(c);
B:=[x : x in F | x ne 0];

Exp:=function(qq,g,t)
R<x>:=TwistedPolynomials(F:q:=qq);
f:=R!g; n:=Degree(f)-1; 
repeat n:=n+1;
v:=[t];
for i in [1..n-1] do
v:=v cat [0];
end for;
s:=v cat [1]; 
g:=R!s; _,r:=Quotrem(g,f); 
until r eq R![0] or n ge c+2;
if r eq R![0] then
return n;
end if;
if n ge c+2 then
return c^Degree(f);
end if;
end function;

EXP:=function(aa,h) 
BB:=[x : x in F | x ne 0]; C:={};
for j in BB do
C:= C join {Exp(aa,h,j)};
end for;
return Minimum(C);
end function;
\end{verbatim}
\end{prog}

\begin{ex}
Let $\mathbb{F}_q:=\{0,1,w,w^2,\dots ,w^{q-2}\}$ be a field with $q$ elements. 
By using Program \ref{Prog-2} and Corollary \ref{existence cor}, we can deduce the
existence Table \ref{Table 3} of some MDS skew $(\alpha,\theta)$-cyclic
$[n,k]_q$-codes whose parameters are not met in the commutative case
by any MDS $\alpha$-constacyclic code (compare this table with Tables $1, 2$ and $3$ in \cite{M3}):  
{\tiny
\begin{table}[!h]
\begin{center}\renewcommand{\arraystretch}{1.6}
\setlength{\tabcolsep}{8pt}
\begin{tabular}{ c c p{1cm} c c c }\hline
$q$ & $n$ & $k$ & $\alpha$ & $p^t$ for $\theta$ & Polynomial $x^2+ax+b$ \\ \hline
$8$ & $6$ & $3, 4$ & $1,1$ & $2$ & $x^2+x+w^3$ \\
$9$ & $6$ & $3, 4$ & $1,w^4$ & $3$ & $x^2+x+w^2$ \\
$16$ & $8$ & $3, 4, 5$ & $1,1,1$ & $2$ & $x^2+x+w$ \\
$16$ & $12$ & $3, 4, 5$ & $1,1,1$ & $2$ & $x^2+wx+w$ \\
$25$ & $10$ & $3, 4, 5$ & $w^{12},w^6,1$ & $5$ & $x^2+x+w^2$ \\
$32$ & $10$ & $3, 4, 5$ & $1,1,1$ & $2$ & $x^2+x+w$ \\
$32$ & $15$ & $3, 4, 5$ & $1,1,1$ & $2$ & $x^2+x+w^7$ \\
$49$ & $14$ & $3, 4, 5$ & $w^{32},w^{24},w^{16}$ & $7$ & $x^2+x+w^2$ \\
$64$ & $12$ & $3, 4, 5$ & $1,1,1$ & $2$ & $x^2+x+w$ \\
$64$ & $18$ & $3, 4, 5$ & $1,1,1$ & $2$ & $x^2+x+w^3$ \\[2pt]
\hline \\
\end{tabular}
\caption{Existence of some MDS skew $(\alpha,\theta)$-cyclic
$[n,k]_q$-code.}
\label{Table 3}
\end{center}
\end{table}} \\
For instance, consider the first case $q=8$ in Table \ref{Table 3} and, for simplicity, let $k=3$. When $\theta=id$, the only possible exponents less or equal to $q+1=9$ of polynomial
of type $x^2+ax+b$ with $b\neq 0$ are $2, 3, 7$ and $9$. On the other hand, Table \ref{Table 3} shows that one can construct
a MDS skew $(1,\theta)$-cyclic $[6,3]_8$-code. Indeed, by using the same notations as in 
Corollary \ref{existence cor}, we have
$$A:=\begin{pmatrix}
0 & 1\\
-w^3 & -1
\end{pmatrix}=\begin{pmatrix}
0 & 1\\
w^3 & 1
\end{pmatrix}\ \ \mathrm{and} \ \ \nu_3[x_0:x_1]=[x_0^2:x_0x_1:x_1^2]\ .$$
Thus we obtain the set 
$$\left\{ [(1,0)], [(0,1)], [(w^3,1)], [(w^3,w^2)], [(1,w^3)], 
[(1,1)] \right\}\subset\mathbb{P}^1(\mathbb{F}_8)$$ which gives via $\nu_3$ the following
set of points in $\mathbb{P}^2(\mathbb{F}_8)$: 
$$\mathcal{K}:=\left\{ [(1,0,0)], [(0,0,1)], [(w^6,w^3,1)], [(w^6,w^5,w^4)], [(1,w^3,w^6)], 
[(1,1,1)] \right\}\ .$$
By $\mathcal{K}\subset\mathbb{P}^2(\mathbb{F}_8)$ one can construct the following matrix
$$H=\begin{pmatrix}
1 & 0 & w^6 & w^6 & 1 & 1\\
0 & 0 & w^3 & w^5 & w^3 & 1\\
0 & 1 & 1 & w^4 & w^6 & 1
\end{pmatrix}$$
which can be interpreted as a parity check matrix of a MDS skew $(1,\theta)$-cyclic
$[6,3]_8$-code whose generator matrix $G$ in standard form is
$$G=\begin{pmatrix}
1 & 0 & 0 & w^5 & w & w^6\\
0 & 1 & 0 & w^2 & w^2 & w^4\\
0 & 0 & 1 & w^5 & w^2 & w^5
\end{pmatrix}\ .$$
Similarly, one can construct all the MDS skew $(\alpha,\theta)$-cyclic
$[n,k]_q$-codes of Table \ref{Table 3}.
\end{ex}

\section{Conclusion}

\noindent In this paper, we consider skew constacyclic codes
$\mathcal{C}$ (called also skew $(\alpha,\theta)$-cyclic codes)
and some algebraic and geometric properties of their dual codes.
After proving again in an easy way that the dual code of
$\mathcal{C}$ is a skew $(\alpha^{-1},\theta)$-cyclic code with an
explicit generator polynomial, we show that the columns of the
generator and the parity check matrices of $\mathcal{C}$ are
orbits of points in projective spaces via semi-linear maps. Two
main applications of this property are given. The first
application consists of some results on $1$-generator skew
quasi-twisted codes which prove that a suitable concatenation of
skew $(\alpha,\theta)$-cyclic codes gives in fact a $1$-generator
skew quasi-cyclic code. The second one shows that under certain
conditions on the parameters $n,k$ and $q$ of a code, the
existence of MDS skew $(\alpha,\theta)$-cyclic $[n, k]_q $-codes
is strictly related to some algebraic conditions which are
explicitly determined.

\bigskip

\section*{Appendix: some bounds for the Hamming 
distance of some skew codes}

A linear code $\mathcal{C}$
of length $n$ and dimension $k$, called an $[n,k]_q$-code, is a $k$-dimensional subspace of $\mathbb{F}_q^n$, where
$\mathbb{F}_q$ is a field with $q$ elements.
Moreover, an $[n,k,d]_q$-code is an $[n,k]_q$-code $\mathcal{C}$ with minimum Hamming distance $d=d(\mathcal{C})$, where
$$d(\mathcal{C}):=\min \{d(\vec{x},\vec{y}):\ \vec{x},\vec{y} \in \mathcal{C}, \ \vec{x}\neq \vec{y} \}$$
and $d(\vec{x},\vec{y})$ is the Hamming distance between the two vectors $\vec{x},\vec{y}$. 

Let us denote by
$d_q(n,k)$ the largest value of $d$ for which an 
$[n,k,d]_q$-code exists. 
It is well known that a large number of new linear codes over 
small fields achieving the best known bounds  
$d_q(n,k)$ have been constructed as cyclic, constacyclic and quasi-cyclic in the non-commutative case 
(see e.g. \cite{AGAS}, \cite{B0}).

In this Appendix, after recalling some basic notions, by using a factorization algorithm of A. Leroy \cite{L}, 
we give in $\S \ref{Sec3-1}$ a BCH lower bound for the minimum Hamming distance $d$ of skew module codes (Theorem \ref{corchete:2}) which can be useful in error-correcting codes and for some decoding algorithms. Furthermore, an immediate application of this general result to skew constacyclic codes is shown (Proposition \ref{corchete:log}) and two upper bounds for $d$ are given in $\S \ref{Sec3-2}$ (see Corollaries \ref{cor1-distance} and \ref{upper d}) as consequences of a remark and of a result which characterizes some Maximum Distance Separable (MDS) skew constacyclic codes (Theorem \ref{teo:ultimos1}).

%

\medskip

With the same notation as in Section $\ref{Sec2}$, consider $f\in R$
of degree $n:=\deg f$ and 
the following $\F_q$-linear isomorphism of
left $\F_q$-modules (or $\F_q$-vector spaces):

\[
\begin{aligned}
\varphi \colon \phantom{aaaaat}\F^n_q\phantom{aaaaa} &\longrightarrow \phantom{aaa} R/Rf\\
    (a_0,a_1,\dots,a_{n-1})&\longmapsto [a_0+a_1x+\cdots a_{n-1}x^{n-1}]\ .
\end{aligned}
\]  

\smallskip

\noindent Observe that $f$ can be assumed to be a monic polynomial in $R$ and let $m$ be the order of $\theta$. 
Without any condition on $m$, note that the set $R/Rf$ can
be always considered also as a left
$R$-module.

\medskip

Let us recall here the definitions of the main skew codes
we will treat in this Appendix.

\begin{defi}[\cite{B1}, \cite{TT1}]\label{def}
An $f$-module $\theta$-code $\mathscr{C}$
is a linear code in $\F_q^n$ which corresponds via $\varphi$ to a left $R$-submodule
$Rg/Rf\subset R/Rf$ in the basis $1,x,...,x^{n-1}$, where $g$ is a right divisor of $f$ in $R$.
The length of the code $\mathscr{C}$ is $n=\deg f$ and its dimension $k=n-\deg g$.
For simplicity, we will denote this code $\mathscr{C}=(g)_{n,\theta}^k\subset\F_q^n$ and we will
say that $\mathscr{C}$ is generated by $g\in R$. Furthermore, if $\theta=id$ then an $f$-module $\theta$-code is called simply an $f$-module code. 
\end{defi}

By specializing the above general definition, we have another version of Definition \ref{skew GCC} (cf. Theorem \ref{R-module}).

\begin{defi}[\cite{B0},\cite{B1}]
A linear code $\mathscr{C}\subseteq \mathbb{F}_q^n$ is called a
\textit{skew $(\alpha ,\theta )$-cyclic code} (or simply, a
\textit{skew $\alpha$-constacyclic code}) 
if $\alpha\in \F_q^\ast$ and $\mathscr{C}$ is a $(x^n-\alpha)$-module $\theta$-code.
\end{defi}

Finally, we will need also the following basic notion.

\begin{defi}\label{def d}
$($\cite{LLO-0}$)$ For $i\in \mathbb{Z}_{\geq 0}$, we define recursively $\mathcal{N}_i^{\theta}(a)$ for any $a\in\mathbb{F}_q$ as
$$\mathcal{N}_0^{\theta}(a)=1,\ \ \mathcal{N}_{i+1}^{\theta}(a)=\theta (\mathcal{N}_i^{\theta}(a))a\ .$$ Furthermore, if $q(x)=\sum_j a_jx^j\in R$, then by \cite[Proposition 2.9]{LLO-0} the 
skew evaluation of $q(x)$ in $a\in\mathbb{F}_q$ is given by $q(a)=\sum_j a_j\mathcal{N}_j^{\theta}(a)$.
\end{defi}

\subsection{A BCH lower bound for skew module codes}\label{Sec3-1}

In \cite{L}, A. Leroy showed how a factorization of a skew polynomial
can be made in a suitable commutative polynomial ring. By this method, we can transfer for instance some properties of constacyclic codes to skew constacyclic codes and vice versa. More in general, a factorization of a skew polynomial in $\F_q[x;\theta]$ can be made in $\F_q[x]$ by a Leroy's algorithm \cite[Theorem 2.5]{L} and, in the same spirit of Hartmann--Tzeng bounds \cite{HT}, a lower bound for 
the minimum Hamming distance can be found. 

For any prime $p$ and $i\in\mathbb{Z}_{\geq 1}$, A. Leroy defines $[i]:=\dfrac{p^i-1}{p-1}$ and
\[  \F_q[x^{[]}]:=\left\{ \sum_{i = 0}^{m}\alpha_i x^{[i]}\in\F_q[x] \ : \ m\in\mathbb{Z}_{\geq 0} \right\}.  \]
\noindent Let us note here that the Leroy's algorithm and the above definitions can be generalized by considering a power of the Frobenius automorphism. More precisely, consider the automorphism $\theta(a):=a^{p^s}$ for all $a\in\F_q$ with $1\leq s\leq r$.
By putting $[i]_s:=\dfrac{(p^s)^i-1}{p^s-1}$ instead of $[i]$, the set $\F_q[x^{[]}]$ can be replaced with
\[  \F_q[x^{[]_s}]:=\left\{ \sum_{i\geq 0}\alpha_i x^{[i]_s}\in\F_q[x] \ : \ m\in\mathbb{Z}_{\geq 0} \right\}\subseteq \F_q[x] .  \]
For any polynomial $p(x):=\sum_{i=0}^{m}\alpha_ix^i\in R$, one can define
$$p^{[]_s}(x):=\sum_{i=0}^{m}\alpha_ix^{[i]_s}\in\F_q[x^{[]_s}]$$ the $[p^s]-$polynomial associated to $p(x)$ and $\F_q[x^{[]_s}]$
the set of $[p^s]-$polynomials. 

Moreover, note that $[i]_s=[i]$ and $\F_q[x^{[]_s}]=\F_q[x^{[]}]$ for $s=1$ and observe that a similar 
result as \cite[Theorem 2.5, (5)]{L} holds, that is, for any $b(x)\in R$ we have
\begin{displaymath}
\qquad \quad t(x)\in Rb(x) \iff t^{[]_s}(x)\in\F_q[x]b^{[]_s}(x)\ . \tag{*}
\end{displaymath}

\medskip

Putting $R':=\F_q[x]$, we can consider the following 
$\F_q$-linear isomorphism $\varphi^{[]_s}$ of (left) $\F_q$-vector spaces associated to $\varphi$:
\[
\begin{aligned}
\varphi^{[]_s} \colon \phantom{aaaaat}\F^{[n]_s}_q\phantom{aaaaa} &\longrightarrow \phantom{aaa} R'/R'f^{[]_s}\\
    (b_0,b_1,\dots,b_{[n]_s-1})&\longmapsto [b_0+b_1x+\cdots + b_{[n]_s-1}x^{[n]_s-1}]\ .
\end{aligned}
\]  

\smallskip

\noindent Moreover, note that the following $\F_q$-linear morphism of left $\F_q$-vector spaces
\[
\begin{aligned}
\psi_s \colon \phantom{aaaaat}R\phantom{aaaaa} &\longrightarrow \phantom{aaa} R'\\
    a(x):=a_0+a_1x+...+a_{m}x^{m} & \longmapsto a^{[]_s}(x):=a_0+a_1x^{[1]_s}+...+a_{m}x^{[m]_s}
\end{aligned}
\]  

\smallskip

\noindent allows us to define by $(*)$  an injective $\F_q$-linear morphism $i$ between left $\F_q$-vector spaces
{\small
\[
\begin{aligned}
i : \phantom{aat}R/Rf\phantom{aaa} &\longrightarrow \phantom{aaa} R'/R'f^{[]_s}\\
    [a(x)] \phantom{aaat} & \longmapsto \phantom{aaa} i([a(x)]):=
    [a^{[]_s}(x)]
\end{aligned}
\]  
}

\noindent This gives the following commutative diagram:

\begin{displaymath}
    \xymatrix{
        \F_q^n \ar[r]^{j} \ar[d]_{\varphi} & \ \F_q^{[n]_s} \ar[d]^{\varphi^{[]_s}}\\
        R/Rf \ \ar[r]^{i}      & \ R'/R'f^{[]_s} } \tag{$**$ }
\end{displaymath}

\noindent where $j:=(\varphi^{[]_s})^{-1}\circ i\circ\varphi$ is an injective $\mathbb{F}_q$-linear 
morphism of left $\mathbb{F}_q$-vector spaces. 

\begin{obs}\label{keyremark}
From $(**)$ it follows that $wt_{[n]_s}(j(\vec{c}))=wt_n(\vec{c})$ for any $\vec{c}\in\F_q^n$, where $wt_m(\vec{e})$ 
denotes the weight of $\vec{e}\in\F_q^m$.
\end{obs}

\begin{defi}
Let $\mathscr{C}:=(g)_{n,\theta}^k\subset\F_q^n$ be an $f$-module $\theta$-code generated by $g\in R$. 
Then $\mathscr{C}^{[]_s}$ denotes the linear code given by $(\varphi^{[]_s})^{-1}(g^{[]_s})_{[n]_s}^h$ with 
$h:=[n]_s-[\deg g]_s$, i.e. $\mathscr{C}^{[]_s}=(g^{[]_s})_{[n]_s, id}^{[n]_s-\deg g^{[]_s}}$, where $\deg g^{[]_s}=[\deg g]_s$.
\end{defi}

\begin{prop}
Let $\mathscr{C}\subseteq\F_q^n$ be a skew $(\alpha,\theta)$-cyclic code for some $\alpha\in\F_q^\ast$.
Then $\mathscr{C}^{[]_s}$ 
is an $\alpha$-constacyclic code in $\F_q^{[n]_s}$.
\end{prop}

\begin{proof}
By Theorem \ref{GMaruta:1}, let $g\in R$ be the generator polynomial of $\mathscr{C}$.
Since $g$ is a right divisor of $x^n-\alpha$, by $(*)$ we see that the $[p^s]-$polynomial $g^{[]_s}\in\F_q[x^{[]_s}]$ associated to $g$
is a right divisor of $x^{[n]_s}-\alpha$. Therefore, the linear code  $\mathscr{C}^{[]_s}:=(\varphi^{[]_s})^{-1}(g^{[]_s})_{[n]_s}^{[n]_s-[\deg g]_s}$ is an $\alpha$-constacyclic code in $\F_q^{[n]_s}$.
\end{proof}

\begin{obs}\label{remark}
More generally, if $\mathscr{C}\subseteq\F_q^n$ is an $f$-module $\theta$-code generated by $g\in R$,
then still using $(*)$, we see that the linear code $\mathscr{C}^{[]_s}$ 
is an $f^{[]_s}$-module code in $\F_q^{[n]_s}$ generated by $g^{[]_s}\in R'$. Furthermore, it is known that
$\mathscr{C}^\perp$ is again a module $\theta$-code if and only if $\mathscr{C}$ is a skew constacyclic code
(see \cite[Theorem 1]{B1}).
\end{obs}

Let us prove now the following general result.

\begin{thm}\label{corchete:2}
Let $\mathscr{C}\subseteq\F_q^n$ be an $f$-module $\theta$-code generated by $g\in R$. Then 
$$d(\mathscr{C})\geq d(\mathscr{C}^{[]_s})$$ and equality holds if and only if there exists 
a vector $\vec{v}\in\mathscr{C}^{[]_s}$ of minimal weight such that $\varphi^{[]_s}(\vec{v})=[v^{[]_s}]$ with 
$v^{[]_s}\in\F_q[x^{[]_s}]$.
\end{thm}

\begin{proof}
Note that from $(**)$ it follows that $j(\mathscr{C})$ is a linear code in $\mathbb{F}_q^{[n]_s}$ such that $j(\mathscr{C})\subseteq\mathscr{C}^{[]_s}$. Moreover, by Remark \ref{keyremark}
we have $wt_n(\vec{v})=wt_{[n]_s}(j(\vec{v}))$ for any $\vec{v}\in\mathscr{C}$, that is, $d(\mathscr{C})=d(j(\mathscr{C}))$.

Therefore, we get $d(\mathscr{C})=d(j(\mathscr{C}))\geq d(\mathscr{C}^{[]_s})$ and by $(*)$ we conclude that
$$d(\mathscr{C}) = d(\mathscr{C}^{[]_s}) \iff d(j(\mathscr{C})) = d(\mathscr{C}^{[]_s}) \iff \exists\ \vec{c}\in \mathscr{C}\ :\ wt_{[n]_s}(j(\vec{c}))=d(\mathscr{C}^{[]_s}) $$
i.e. $d(\mathscr{C}) = d(\mathscr{C}^{[]_s}) \iff \exists\ \vec{v}\in \mathscr{C}^{[]_s} \ :\ wt_{[n]_s}(\vec{v})=d(\mathscr{C}^{[]_s})\ \mathrm{and}\ \vec{v}\in j(\mathscr{C})$.
\end{proof}

\begin{prop}
Let $\mathscr{C}\subseteq\F_q^n$ be an $f$-module $\theta$-code generated by $g\in R$ with deg$(g)>0$. 
Then
$\mathscr{C}^{[]_s}\subset\F_q^{[n]_s}$ is not a Maximum Distance Separable (MDS) linear code.
\end{prop}

\begin{proof}
Suppose that $\mathscr{C}^{[]_s}\subset \F_q^{[n]_s}$ is a MDS linear code. Then
$$
d(\mathscr{C}^{[]_s})=[n]_s-\text{dim}(\mathscr{C}^{[]_s})+1
=\text{deg} (g^{[]_s})+1=[ \text{deg}(g) ]_s+1,
$$
where $g^{[]_s}$ is the generator polynomial of $\mathscr{C}^{[]_s}$.
So by Theorem \ref{corchete:2} we conclude that
$$[ \text{deg}(g) ]_s+1=d(\mathscr{C}^{[]_s})\leq d(\mathscr{C})\leq n-\dim \mathscr{C}+1=\text{deg}(g)+1,$$ 
but this gives a contradiction because deg$(g)\geq 1$ and $[l]_s>l$ for all $l\in\Z_{\geq 1}$. 
\end{proof}

\begin{lem}\label{Lema[k]}
Let $\F_{q^{[k]_s}}^*=\langle w\rangle$ and $\F_{q^{[k]_s}} \subseteq \F_{q^{[n]_s}}$. Then 
$$[n]_s\leq q^{[k]_s}-1\ \iff \ rank\left( V(w,w^2,\dots,w^{[n]_s-1}) \right)=[n]_s,$$
where
$$V(w,w^2,\dots,w^{[n]_s-1}) :=
\begin{pmatrix}
1&1&1&\cdots&1\\
1&w&w^2&\cdots&w^{[n]_s-1}\\
1&w^2&\left(w^2\right)^2&\cdots&\left(w^{[n]_s-1}\right)^2\\
\vdots&\vdots&\vdots& &\vdots\\
1&w^{[n]_s-1}&\left(w^2\right)^{[n]_s-1}&\cdots&\left(w^{[n]_s-1}\right)^{[n]_s-1}\\
\end{pmatrix}.$$
\end{lem}

\begin{proof} 
Write $\mathcal{M}:=V(w,w^2,\dots,w^{[n]_s-1})$.

\noindent ``$\Rightarrow$''\ Suppose that $rank\left( \mathcal{M} \right)\neq [n]_s$. Then we have 

\begin{displaymath}
det(\mathcal{M}):=\displaystyle\prod\limits_{0\leq i< j\leq [n]_s-1}(w^j-w^i)=0, \tag{$\#$}
\end{displaymath}

that is, $w^{j-i}=1$ for some $i,j$ such that $0\leq i< j\leq [n]_s-1$. So the order of $w$,
denoted by $ord(w)$, must be less than or equal to $[n]_s-1$. Thus we get $[n]_s\leq q^{[k]_s}-1= ord(w)\leq [n]_s-1$, which is impossible. Hence $det(\mathcal{M})\neq 0$, that is, $rank\left( \mathcal{M} \right) = [n]_s$.

\medskip

\noindent ``$\Leftarrow$''\ Note that $rank\left( \mathcal{M} \right)=[n]_s$ implies by $(\#)$ that $w^{j-i}\not=1$ for any $i,j$ such that $0\leq i< j\leq [n]_s-1$. Thus $q^{[k]_s}-1=ord(w)> [n]_s-1$, i.e. $[n]_s\leq q^{[k]_s}-1$.
\end{proof}

By way of example, let us give here an application of the above arguments to obtain a lower bound for the minimum Hamming distance of skew constacyclic codes by finding roots of polynomials in $R'$ instead of $R$.

\begin{prop}\label{corchete:log}
Let $\mathscr{C}\subseteq\F_q^n$ be an $f$-module $\theta$-code generated by $g\in R$ of degree $k$.
Let $\F_{q^{[k]_s}}^*=\langle w\rangle$ and $\F_{q^{[k]_s}} \subseteq \F_{q^{[n]_s}}$ with $q=p^r$ for some prime $p\geq 2$. Assume that
\[ k\geq \log_{p^s}\left[1+\left(\frac{p^s-1}{r}\right)\cdot \log_p\left(\frac{(p^s)^n+p^s-2}{p^s-1}\right)\right]\ . \]
If $g^{[]_s}(\omega^{l+ci})=0$ for $i=0,\dots,\Delta-2$, some $l\in\Z_{\geq 0}$
and $c$ such that $(c,q^{[k]_s}-1)=1$, then $d(\mathscr{C})\geq \Delta$.
\end{prop}

\begin{proof}
Note that
$$
k\geq \log_{p^s}\left[1+\left(\frac{p^s-1}{r}\right)\cdot \log_p\left(\frac{(p^s)^n+p^s-2}{p^s-1}\right)\right]
\Longleftrightarrow $$
$$ r\left(\frac{(p^s)^k-1}{p^s-1}\right)\geq \log_p\left(1+[n]_s \right)
\Longleftrightarrow
p^{r\cdot [k]_s} \geq \left(1+[n]_s \right)\Longleftrightarrow
q^{[k]_s}-1 \geq [n]_s\ .$$

\smallskip

\noindent Thus $ord(\omega)=q^{[k]_s}-1\geq [n]_s$. Furthermore, since $(c,q^{[k]_s}-1)=1$, we have
\[ (\omega^c)^j\not= 1, \text{ for }j=1,\dots,[n]_s-1\ . \]
From Remark \ref{remark} we know that $\mathscr{C}^{[]_s}\subseteq\F_q^{[n]_s}$ is an $f^{[]_s}$-module code with generator polynomial  $g^{[]_s}\in R'$.
Since $g^{[]_s}(\omega^{l+ci})=0$, for $i=0,\dots,\Delta-2$ and some $l\in\Z_{\geq 0}$, by Theorem \ref{corchete:2} and
\cite[Theorem 3.9]{TT2} we conclude that $d(\mathscr{C})\geq d(\mathscr{C}^{[]_s})\geq \Delta$.
\end{proof}

\begin{obs}
In the same spirit of Proposition \ref{corchete:log}, any result which gives a lower bound for the minimum Hamming distance
$d(\mathscr{C}^{[]_s})$ of $\mathscr{C}^{[]_s}$ in the commutative case can be used to get a lower bound for $d(\mathscr{C})$
by Theorem~\ref{corchete:2}.
\end{obs}

\subsection{Upper bounds for some skew 
module codes}\label{Sec3-2}

First of all, let us give here the following

\begin{obs}\label{rem upper bound}
If $\mathscr{C}\subseteq \F^n_q$ is a linear code, then $d(\mathscr{C})+d(\mathscr{C}^\perp)\leq (n-k+1)+[n-(n-k)+1]=n+2$. 
So, if we know that $d(\mathscr{C}^\perp)\geq \delta$ for some $\delta\in\mathbb{Z}_{\geq 1}$, then we have 
$d(\mathscr{C})\leq \min \{n-\delta+2,n-k+1\}$, which gives an easy  
upper bound for $d(\mathscr{C})$.
\end{obs}

\noindent From Theorem \ref{GMaruta:1e} we know that the dual code $\mathscr{C}^\perp$ of a skew $(\alpha,\theta)$-cyclic code is a skew $(\alpha^{-1},\theta)$-cyclic code. The first upper bound for the minimum Hamming distance is obtained
by the following result.

\begin{cor}\label{cor1-distance}
Let $\mathscr{C}\subseteq \F^n_q$ be a skew $(\alpha,\theta)$-cyclic code generated by $g\in R$. Then
$$d(\mathscr{C}^{[]_s})\leq d(\mathscr{C})\leq \min \{n-d({\mathscr{C}^\perp}^{[]_s})+2,\deg g+1\},$$
where $\mathscr{C}^{[]_s}\subseteq\F_q^{[n]_s}$ is an $\alpha$-constacyclic code generated by $g^{[]_s}\in R'$ and
${\mathscr{C}^\perp}^{[]_s}\subseteq\F_q^{[n]_s}$ is an $\alpha^{-1}$-constacyclic code generated by the $[p^s]$-polynomial 
$h^{[]_s}\in R'$ associated to $h\in R$ given by {\em Proposition \ref{asterisco}}.
\end{cor}

\begin{proof}
The statement follows from Theorem \ref{GMaruta:1e}, Remark \ref{rem upper bound} and Theorem \ref{corchete:2} applied to $\mathscr{C}$ and 
$\mathscr{C}^\perp$.
\end{proof}

Before to give another upper bound of the minimum Hamming
distance for some special skew module codes (e.g., see \cite{BL}), we prove the
following result which extends \cite[Theorem 2.10]{M2} to the
non-commutative case.

\begin{prop}\label{teo:ultimos1}
Let $q=p^r$ for some prime $p$ and $r\in\mathbb{Z}_{\geq 1}$. Let $\mathscr{C}\subseteq\F_q^n$ be an $f$-module $\theta$-code generated by $g(x)\in R$. 
Suppose that $\deg g(x)=k$ with $0<k<n$ and that there exist $\eta\in\mathbb{F}_{q^k}^*$ and $t\in\mathbb{Z}$ with $1\leq t\leq r$ such that
$$g(x)=\text{l.l.c.m.}\left\{ x-\eta^{p^{it}}\colon i=0,\dots,k-1\right\}\in R\ .$$
Then 

\medskip

$\mathscr{C}$ is a MDS code $\iff$ $\text{det}\left[ \mathcal{N}_{i_j}^\theta \left(\eta^{p^{(l-1)t}}\right)\right]_{1\leq j,l\leq k} \not=0$, 
for any subset of indexes $\{ i_1,\dots, i_k\}\subset\{0,1,2,\dots,n-1\}$ with
$i_1<\dots < i_k$, where $\mathcal{N}_{i_j}^\theta$ is as in {\em Definition \ref{def d}}.
\end{prop}

\begin{proof} 
Note that, for any $a\in\mathbb{F}_{q^k}$, $\theta$ can be extended to $\mathbb{F}_{q^k}$ by $\theta (a):=a^{p^s}$ and that $N_j^{\theta}(a)=a^{[j]_s}$ for any $j\in\mathbb{Z}_{\geq 0}$. Write
\begin{equation}\label{etaN}
\mathcal{T}_{n,k}:=
\begin{pmatrix}
\mathcal{N}_0^\theta(\eta) 	& \mathcal{N}_1^\theta(\eta)	& \cdots 	& \mathcal{N}_{n-1}^\theta(\eta) \\
\mathcal{N}_0^\theta\left(\eta^{p^t}\right) 	& \mathcal{N}_1^\theta\left(\eta^{p^t}\right)	& \cdots  	& \mathcal{N}_{n-1}^\theta\left(\eta^{p^t}\right) \\
\vdots 						& \vdots 						& 			& \vdots \\
\mathcal{N}_0^\theta\left(\eta^{p^{(k-1)t}}\right)	& \mathcal{N}_1^\theta\left(\eta^{p^{(k-1)t}}\right)	& \cdots 	& \mathcal{N}_{n-1}^\theta\left(\eta^{p^{(k-1)t}}\right) \\
\end{pmatrix}\ ,
\end{equation}
and observe that
\begin{equation}\label{etaNbis}
\mathcal{T}_{n,k}=
\begin{pmatrix}
1 	& \eta	& \eta^{[2]_s} & \cdots 	& \eta^{[n-1]_s} \\
1 	& (\eta)^{p^t} & \left(\eta^{[2]_s}\right)^{p^t} & \cdots  	& \left(\eta^{[n-1]_s}\right)^{p^t} \\
\vdots 						& \vdots 						& 			& \vdots \\
1	& \left(\eta\right)^{p^{(k-1)t}} & \left(\eta^{[2]_s}\right)^{p^{(k-1)t}}	& \cdots 	& \left(\eta^{[n-1]_s}\right)^{p^{(k-1)t}} \\
\end{pmatrix}\ .
\end{equation}
Define $\Theta(z):=z^{p^t}$ for any $z\in\mathbb{F}_{q^k}$ and note that $\Theta^j\left(\eta^{[i]_s}\right)=\left(\eta^{[i]_s}\right)^{p^{jt}}$
for any $i,j\in\mathbb{Z}_{\geq 0}$.
Hence 
\begin{equation}\label{etaNter}
\mathcal{T}_{n,k}=
\begin{pmatrix}
1 	& \eta	& \eta^{[2]_s} & \cdots 	& \eta^{[n-1]_s} \\
\Theta(1) 	& \Theta\left(\eta\right) & \Theta\left(\eta^{[2]_s}\right) & \cdots  	& \Theta\left(\eta^{[n-1]_s}\right) \\
\vdots 						& \vdots 						& 			& \vdots \\
\Theta^{k-1}(1)	& \Theta^{k-1}\left(\eta\right) & \Theta^{k-1}\left(\eta^{[2]_s}\right)	& \cdots 	& \Theta^{k-1}\left(\eta^{[n-1]_s}\right) \\
\end{pmatrix}\ ,
\end{equation}
and by \eqref{etaN}, \eqref{etaNbis} and \cite[Corollary 4.13]{LLO-0}, we deduce that

\medskip

\noindent $\text{det}\left[ \mathcal{N}_{i_j}^\theta \left(\eta^{p^{(l-1)t}}\right)\right]_{1\leq j,l\leq k}=0$,\ \ \ for some $\{i_1,\dots ,i_k\}\subseteq
\{0,\dots,n-1\}\iff$ 

\medskip

\noindent $\iff\exists\vec{c}\in\mathbb{F}_q^n\ : \ \mathcal{T}_{n,k}\cdot {}^t\vec{c}={}^t\vec{0}\ ,\ wt(\vec{c})\leq k\iff \exists
\vec{c}\in\mathscr{C}\ : \ wt(\vec{c})\leq k\iff$

\medskip

\noindent $\iff\mathscr{C}$ is not a MDS code, since $\dim\mathscr{C}=n-\deg g(x)=n-k$.
\end{proof}

Finally, a second upper bound for the minimum Hamming distance 
of some skew module codes can be given as a consequence of the proof of Proposition \ref{teo:ultimos1}.

\begin{cor}\label{upper d}
Let $q=p^r$ for some prime $p$ and $r\in\mathbb{Z}_{\geq 1}$. Let $\mathscr{C}\subseteq\F_q^n$ be an $f$-module $\theta$-code generated by $g(x)\in R$. 
Suppose that $\deg g(x)=k$ with $0<k<n$ and that there exist $\eta\in\mathbb{F}_{q^k}^*$ and $t\in\mathbb{Z}$ with $1\leq t\leq r$ such that
$$g(x)=\text{l.l.c.m.}\left\{ x-\eta^{p^{it}}\colon i=0,\dots,k-1\right\}\in R\ .$$
Then 
$$ d(\mathscr{C})\leq \min\left\{ \mathrm{rank}\left( \mathcal{M}_{i_1,\dots,i_k} \right) \colon 1\leq i_1 <\cdots<i_k\leq n\right\}+1\ , $$
where $\mathcal{M}_{i_1,\dots,i_k}$ is a $k\times k$ submatrix of the matrix $\mathcal{T}_{n,k}$ defined by $(\ref{etaN})$.
\end{cor}

\begin{proof}
With the same notation as in the proof of Proposition \ref{teo:ultimos1}, we have 
\begin{equation}\label{eq}
\vec{c}\in\mathscr{C}\ \iff\ \vec{c}\in\mathbb{F}^n_q\ \ \ \mathrm{and}\ \ \ \mathcal{T}_{n,k}\cdot {}^t\vec{c}={}^t\vec{0}\ .
\end{equation}
Let $\delta:=\min\left\{ \text{rank}\left( \mathcal{M}_{i_1,\dots,i_k} \right) \colon 1\leq i_1 <\cdots<i_k\leq n\right\}$ and note that $\delta\leq k$.
If $\delta=k$, then $d(\mathscr{C})\leq k+1=\delta+1$, i.e. $d(\mathscr{C})\leq\delta+1$. 

Thus, assume that $\delta\leq k-1$ and let $\mathcal{M}_{j_1,\dots,j_k}$ be a submatrix 
of $\mathcal{T}_{n,k}$ such that $\delta=\text{rank}\left( \mathcal{M}_{j_1,\dots,j_k} \right)$.
In particular, this gives
$$\det 
\begin{pmatrix}
\eta^{[j_1]_s} & \cdots 	& \eta^{[j_{\delta+1}]_s} \\
\Theta\left(\eta^{[j_1]_s}\right) & \cdots  	& \Theta\left(\eta^{[j_{\delta+1}]_s}\right) \\
\vdots 						& 			& \vdots \\
\Theta^{\delta}\left(\eta^{[j_1]_s}\right)	& \cdots 	& \Theta^{\delta}\left(\eta^{[j_{\delta+1}]_s}\right) \\
\end{pmatrix}
=0\ .$$
By \cite[Corollary 4.13]{LLO-0}, we deduce that $\eta^{[j_1]_s},\dots,\eta^{[j_{\delta+1}]_s}$ are linearly dependent over $\mathbb{F}_q$.
So $(\eta^{[j_1]_s})^{p^{jt}}, \dots, (\eta^{[j_{\delta+1}]_s})^{p^{jt}}$ are linearly dependent over $\mathbb{F}_q$ for any $j,t\in\mathbb{Z}_{\geq 0}$.
Thus it turns out that the columns $j_1,\dots , j_{\delta+1}$ of $\mathcal{T}_{n,k}$ are linearly dependent over $\mathbb{F}_q$, that is,
there exists a non-zero vector 
\[ \vec{c}=(0,\dots,0,c_{j_1},0,\dots,0,c_{j_2},0,\dots,0, c_{j_{\delta+1}},0,\dots,0)\in\F_q^n\]
such that $\mathcal{T}_{n,k}\cdot {}^t\vec{c}={}^t\vec{0}$. Therefore, from \eqref{eq} it follows that there exists a non-zero vector
$\vec{c}\in\mathscr{C}$ with $wt(\vec{c})\leq\delta+1$, i.e. $d(\mathscr{C})\leq\delta+1$.
\end{proof}

\bigskip

\begin{obs}
Let $\mathscr{C}\subseteq \F^n_q$ be a linear $[n,n-k]_q$-code with parity check matrix $H$. From the proof of Corollary \ref{upper d},
one can easily deduce the obvious upper bound
$$d(\mathscr{C})\leq\min\left\{ \text{rank}\left( \mathcal{M}_{k} \right) \colon \mathcal{M}_{k}\ \mathrm{is\ a\ } k\times k\ 
\mathrm{submatrix\ of\ } H \right\}+1\ .$$
\end{obs}

\end{document}